\documentclass[12pt,reqno]{amsart}

\usepackage{times}
\usepackage{graphicx}
\graphicspath{ {./images/} }

\usepackage{hyperref}

\usepackage{amsmath,amsthm,amsfonts,amscd,amssymb}
\usepackage{pst-node}
\usepackage{pst-plot}
\usepackage[all]{xy}
\usepackage{stackengine}

\usepackage{tikz, braids}
\usetikzlibrary{decorations.pathreplacing, arrows}

\usepackage{enumitem,kantlipsum}

\newtheorem{theorem}{Theorem}[section]
\newtheorem{corollary}[theorem]{Corollary}

\newtheorem{lemma}[theorem]{Lemma}
\newtheorem{remark}[theorem]{Remark}
\newtheorem{proposition}[theorem]{Proposition}

\numberwithin{theorem}{section}
\numberwithin{equation}{section}

\newcommand{\ds}{\displaystyle}

\newcommand{\myop}[1]{\operatorname{#1}}

\newcommand{\Tr}{\myop{Tr}} 
\newcommand{\id}{\myop{id}}

\newcommand{\Tt}{\mathcal{T}}

\newcommand{\Real}{\mathbb R}

\newcommand{\set}[1]{\left\{#1\right\}}

\newcommand{\la}{\langle}
\newcommand{\ra}{\rangle}

\newcommand{\Comp}{\mathbb{C}}

\newcommand{\K}{\mathcal{K}}

\newcommand{\Le}{\mathcal{L}}

\newcommand{\U}{\mathcal{U}}

\newcommand{\Om}{\Omega}
\newcommand{\om}{\omega}

\newcommand{\M}[1]{M_{#1}(\mathbb{C})}
\newcommand{\pos}{\mathbf{P}}
\newcommand{\sch}{\mathbf{Sch}}
\newcommand{\BP}{\mathbf{BP}}
\newcommand{\SEP}{\mathbf{SEP}}

\newcommand{\PO}{\mathcal{POS}}
\newcommand{\CP}{\mathcal{CP}}
\newcommand{\SP}{\mathcal{SP}}
\newcommand{\EB}{\mathcal{EB}}

\newcommand{\D}{\mathbf{D}}

\newcommand{\Inv}{{\rm Inv}}
\newcommand{\Cov}{{\rm Cov}}

\newcommand{\Covn}{{\rm Cov}_1}

\begin{document}

\title[$k$-positivity and Schmidt number]{$k$-positivity and Schmidt number under orthogonal group symmetries
}
% or if you want, simply \title{Title of the article}

\author[S.-J. Park]{Sang-Jun Park}
\address{Sang-Jun Park, 
Department of Mathematical Sciences, Seoul National University, 
Gwanak-Ro 1, Gwanak-Gu, Seoul 08826, Republic of Korea}
\email{psj05071@snu.ac.kr}

\author[S.-G. Youn]{Sang-Gyun Youn}
\address{Sang-Gyun Youn, 
Department of Mathematics Education, Seoul National University, 
Gwanak-Ro 1, Gwanak-Gu, Seoul 08826, Republic of Korea}
\email{s.youn@snu.ac.kr }

\maketitle

\begin{abstract}

In this paper, we study $k$-positivity and Schmidt number under standard orthogonal group symmetries. The Schmidt number is a natural quantification of entanglement in quantum information theory. First of all, we exhibit a complete characterization of all orthogonally covariant $k$-positive maps. This generalizes earlier results in \cite{Tom85}. Furthermore, we optimize duality relations between $k$-positivity and Schmidt numbers under compact group symmetries. This new framework enables us to efficiently compute the Schmidt numbers of all orthogonally invariant quantum states.
\end{abstract}

\section{Introduction}

The study of positivity has a rich historical background with significant applications in operator algebra and quantum theory. One prominent application is a characterization of quantum channels, which describe quantum evolutions in open quantum systems interacting with their environments. The mathematical structure of quantum channels is provided by completely positive trace-preserving (CPTP) linear maps in the Schr{\"o}dinger picture. In this paper, we will focus on finite-dimensional quantum systems. Note that, if $d=\min\left\{{\rm dim}(H_A),{\rm dim}(H_B)\right \}$, then a linear map $\Le:B(H_A)\rightarrow B(H_B)$ is CP if and only if $\Le$ is $d$-positive, i.e.
\begin{equation}
    {\rm id}_{d}\otimes \Le: M_d\otimes B(H_A)\rightarrow M_d\otimes B(H_B)
\end{equation}
is positive, where $M_d$ is the set of all linear operators acting on $\Comp^d$. An intermediate concept between positivity and complete positivity is the so-called {\it $k$-positivity} ($1\leq k\leq d$), which means that
\begin{equation}
    {\rm id}_{k}\otimes \Le: M_k\otimes B(H_A)\rightarrow M_k\otimes B(H_B)
\end{equation}
is positive. While extensive research has been conducted on $k$-positivity \cite{TT83, Tom85, CKL92, CK09, JK10, COS18, HLLMH18}, it is hard to expect an efficient method to verify whether a given linear map is $k$-positive or not in general.

On the other hand, quantum entanglement is a fundamental phenomenon in quantum science and serves as a crucial resource in quantum information processing \cite{HHHH09}. Quantifying quantum entanglement is a central issue in this field, and it has been revealed that $k$-positivity is closely linked to quantum entanglement via the so-called \textit{Schmidt number} \cite{TH00}. Specifically, the Schmidt number of a quantum state $\rho$ is greater than $k$ if and only if there exists a certain $k$-positive linear map $\Le$ such that $(\id \otimes \Le)(\rho)$ is not positive semidefinite. In this context, $k$-positive maps can be considered \textit{Schmidt number witnesses} \cite{SBL01}. While various attempts have been made to obtain lower and upper bounds for the Schmidt numbers \cite{TH00, CK09, YLT16, CYT17, HLLMH18, PV19, Car20, LHHGV23}, accurate computations still pose significant challenges. To the best of our knowledge, there are very few explicit examples where the Schmidt numbers can be precisely calculated in high-dimensional systems. Some known examples are the isotropic states \cite{HH99, TH00}, the Werner states \cite{Wer89, Kye23lect}, and recently in \cite{Car20}.

%In this study, building upon our previous work \cite{PJPY23}, we continue to conduct a comprehensive quantitative analysis of $k$-positivity and Schmidt number by focusing on linear maps and bipartite operators with \textit{compact group symmetry}. 

One crucial progress of this paper is a new framework for quantitative analysis of quantum entanglement generalizing the methodologies proposed in \cite{PJPY23}. The main focus of \cite{PJPY23} was to develop a universal framework to study the problem of whether a given state is entangled or not under compact group symmetry. In this paper, we extend their framework to cover more general questions of the Schmidt numbers (Theorem \ref{cor-invSch}). Indeed, we prove that a much smaller set of $k$-positive maps is sufficient as detectors to compute the Schmidt numbers under compact group symmetries. Furthermore, our abstract approach not only establishes the duality between $k$-positive maps and the Schmidt numbers but also provides more general duality results between {\it mapping cones} (Theorem \ref{thm-main1}).

% In fact, we extend our duality result to encompass a more general notion of \textit{mapping cones} (Theorem \ref{thm-main1}) which is introduced by St\o{}rmer \cite{Sto86} and contains many other classes of positive maps (PPT maps and decomposable maps, for example). This extension enables us to compare different classes of positive maps or operators under specific symmetry conditions. For example, although it is already shown that there are always PPT entangled quantum states unless $d_A\cdot d_B\leq 6$, we can again explore the question of PPT entanglement whenever we restrict the class to possess a certain type of invariance, as considered in previous literature (REF: VW01, EW01, CKK+21, SN21, BSGS22, PJPY23). Our result in Corollary \ref{cor-duality problem} addresses general comparison problems using duality and finds several applications. In essence, our framework of group symmetry allows us to approach various mathematical problems related to quantum entanglement, which have already been solved in a general setting, from a fresh perspective.

The generalized duality results enable us to analyze $k$-positivity and the Schmidt number for quantum objects under the standard \textit{orthogonal group symmetries} \cite{VW01} (Theorem \ref{thm-OOkpos}, \ref{thm-OOSch}, and \ref{thm-OOSch2}). Specifically, we provide a complete characterization of $k$-positivity of all linear maps of the form 
\begin{equation}\label{eq03}
   \Le^{(d)}_{a,b}(Z)=(1-a-b)\frac{\text{Tr}(Z)}{d}I_d+aZ+bZ^{\top}
\end{equation}
and apply our duality results to compute the Schmidt numbers of all quantum states of the form
\begin{equation}\label{eq04}
\rho^{(d)}_{a,b}=\frac{1-a-b}{d^2}I_{d^2} +\frac{a}{d}\sum_{i,j=1}^d |ii\ra\la jj|+\frac{b}{d}\sum_{i,j=1}^d |ij\ra\la ji|.
\end{equation}
Note that \cite{Tom85,TH00} cover special cases $\Le^{(d)}_{a,0}$, $\Le^{(d)}_{0,b}$, and $\rho^{(d)}_{a,0}$ where all these subclasses are parametrized as $1$-dimensional spaces. 
% {\color{blue}Furthermore, this gives the partial answer to the question raised in \cite{SP18}.}
To the best of our knowledge, our computations provide the first example of the complete characterization of Schmidt numbers in a non-trivial class parameterized by at least two real variables (in arbitrarily high dimensional settings).

% Not only the class \eqref{eq04} includes both isotropic states and Werner states \cite{Wer89, HH99}, our computations provide the first example of the complete characterization of Schmidt numbers in a {\color{purple}non-trivial} class parameterized by at least two real variables (in arbitrarily high dimensional settings), to the best of our knowledge.

To visualize the full characterization of the $k$-positivity of $\Le^{(d)}_{p,q}$ and the Schmidt number of $\rho^{(d)}_{a,b}$, let us denote by $\mathcal{P}^{(d)}_k$ the set of $k$-positive linear maps $\Le^{(d)}_{p,q}$, and by $\mathcal{S}^{(d)}_k$ the set of quantum states $\rho^{(d)}_{a,b}$ whose Schmidt number is less than or equal to $k$. Our main results reveal that the geometic structures of $\mathcal{P}^{(d)}_k$ and $\mathcal{S}^{(d)}_k$ are categorized into four distinct cases: (1) $k=1$, (2) $1<k\leq \frac{d}{2}$, (3) $\frac{d}{2}<k<d$, (4) $k=d$. For a special case $d=4$, we provide a visual representation of the convex sets $\mathcal{P}^{(4)}_k$ and $\mathcal{S}^{(4)}_k$ ($1\leq k\leq 4$) in Figure \ref{fig-OO4pos} below. Note that the regions are highly non-trivial to describe since conics are necessary for both $\mathcal{P}^{(4)}_3$ and $\mathcal{S}^{(4)}_3$.

\begin{figure}[htb!] 
    \centering
    \includegraphics[scale=0.37]{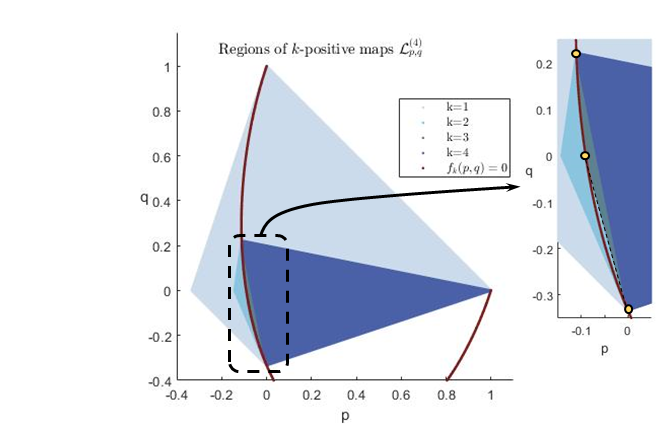}    \includegraphics[scale=0.30]{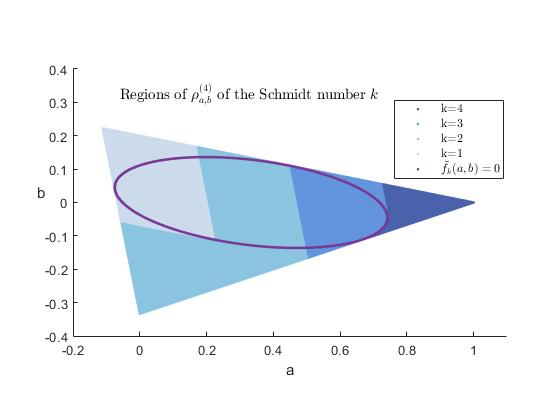}
    \caption{Regions of the $k$-positive maps $\mathcal{L}_{p,q}^{(4)}$ and quantum states $\rho^{(4)}_{a,b}$ of the Schmidt number $k$}
    \label{fig-OO4pos}
\end{figure}

\section{Preliminaries}

\subsection{Positive maps and Schmidt number} \label{sec-notations}

Let us fix some notations that are frequently used throughout this paper. First of all, we will use bracket notation from physics. All Hilbert space $H$ will be assumed to be finite-dimensional equipped with an inner product $\la \cdot |\cdot \ra$ and an orthonormal basis $\set{e_i}_{i=1}^d$. For $v,w\in H$, we denote by $|v\ra$ an operator mapping $\lambda\in \Comp$ to $\lambda v\in H$, and by $\la w | $ a linear functional mapping $\xi\in H$ to $\la w | \xi\ra\in \Comp$. Then $|v\ra\la w|$ is a rank one operator which maps $\xi\in H$ to $\la w|\xi\ra  |v\ra\in H$. In particular, we write $|i\ra:=|e_i\ra$ and $|ij\ra:=|i\ra\otimes |j\ra$. Then $\set{|i\ra\la j|}_{i,j=1}^d$ forms matrix units in $B(H)$. An operator $\rho\in B(H)$ is called a \textit{quantum state} if it is positive and $\Tr(\rho)=1$. We denote by $\D(H)$ the set of all quantum states.

For every vector $\xi\in H_A\otimes H_B$, there exists a Schmidt decomposition \cite{NiCh} $|\xi\ra=\sum_{i=1}^k{\lambda_i}|v_i\ra\otimes |w_i\ra$ where $\lambda_1\geq \cdots \geq \lambda_k>0$, and $\set{v_i}_{i=1}^k$ and $\set{w_i}_{i=1}^k$ are orthonormal subsets in $H_A$ and $H_B$, respectively. The numbers $k$ and $\set{\lambda_i}_{i=1}^k$ are uniquely determined, and  we call $k$ the \textit{Schmidt rank} of $\xi$ and write ${\rm SR}(|\xi\ra)=k$. Now we denote by $\pos_{AB}$ the set of positive operators on $H_A\otimes H_B$ and consider the following subsets
    $$\sch_{k,AB}:={\rm conv}\set{|\xi\ra\la \xi|: {\rm SR}(|\xi\ra)\leq k}$$
for any natural number $k$ (or simply write $\pos$ and $\sch_k$ when these cause no confusion). Then the \textit{Schmidt number} of a positive operator $X\in \pos_{AB}$ is defined as the smallest natural number $k$ such that $X\in \sch_k$, and we write ${\rm SN}(X)=k$. Note that ${\rm SN}(X)\leq \min(d_A,d_B)$ and $\sch_k=\pos_{AB}$ for all $k\geq \min(d_A,d_B)$. We employ another notation $\SEP_{AB}:=\sch_{1,AB}$ for the positive operators of the Schmidt number $1$, and $X\in \pos_{AB}$ is called \textit{separable} if $X\in \SEP_{AB}$ and \textit{entangled} otherwise.

Let us denote by $B(B(H_A),B(H_B))$ the set of all linear maps from $B(H_A)$ into $B(H_B)$ and by $B^h(B(H_A),B(H_B))$ the set of all \textit{Hermitian preserving maps}, i.e., $\Le\in B(B(H_A),B(H_B))$ with $\Le(Z)^*=\Le(Z^*)$ for $Z\in B(H_A)$. We also denote by $\PO_{AB}\subset B^{h}(B(H_A),B(H_B))$ the cone of positive maps from $B(H_A)$ into $B(H_B)$. A list of subclasses of positive maps of our interest is the following:
\begin{itemize}
    \item $\PO_{k,AB}$, the set of \textit{$k$-positive} maps (note that $\PO_1=\PO$),

    \item $\CP_{AB}$, the set of \textit{completely positive (CP)} maps, 

    \item $\SP_{k,AB}:={\rm conv}\set{{\rm Ad}_K: K\in B(H_B,H_A), {\rm rank}(K)\leq k}$, the set of \textit{$k$-superpositive} maps \cite{SSZ09}, where ${\rm Ad}_K(X):=KXK^*$,

    \item $\EB_{AB}:=\SP_1$, the set of \textit{entanglement-breaking (EB)} maps \cite{HSR03}.

    % \item $\DEC_{AB}:=\CP_{AB} + (\top_B\circ \CP_{AB})$, the set of \textit{decomposable} maps where $\top_B$ is the transpose map on $B(H_B)$,

    % \item $\PPT_{AB}:=\CP_{AB} \cap (\top_B\circ \CP_{AB})$, the set of \textit{PPT} maps,
\end{itemize}
Note that we have a nested chain of the subclasses as follows. 
\small
\begin{align}
    \notag \PO\supsetneq \PO_2\supsetneq \cdots \supsetneq \PO_{\min(d_A,d_B)}=\CP& \\
    =\;\SP_{\min(d_A,d_B)} & \supsetneq  \cdots \supsetneq\SP_2\supsetneq \EB. \label{eq-chain1}
    % \\
    % \PO\supset \DEC \supsetneq \;&\CP \supsetneq \PPT \supset \EB. \label{eq-chain2}
\end{align}
\normalsize

% Moreover, the two inclusions $\PO\supset \DEC$ and $\PPT\supset \EB$ is in general strict unless $(d_A,d_B)=(2,2),(2,3),(3,2)$ \cite{Cho75b, Wor76, HHH96, Hor97}.
% Positivity of linear maps has been compared with decomposability. A map $\Le\in B(B(H_A),B(H_B))$ is called decomposable if we can write $\Le=\Phi_1+\top_B\circ \Phi_2$ for some CP maps $\Phi_1,\Phi_2$. Every positive map is decomposable when $(d_A,d_B)=(2,2), (2,3), (3,2)$, but it is not true for other dimensions (REF:Wor76).

On the other hand, linear maps acting on quantum systems are standardly identified with bipartite operators via the so-called \textit{Choi-Jamio\l{}kowski correspondence} \cite{Jam72,Cho75a}. For $\Le\in B(B(H_A),B(H_B))$, the \textit{(normalized) Choi matrix} $C_{\Le}\in B(H_A\otimes H_B)$ is defined by
    $$C_{\Le}:=(\id_A\otimes \Le)(|\Om_A\ra\la \Om_A|)=\frac{1}{d_A}\sum_{i,j=1}^{d_A}|i\ra\la j|\otimes \Le(|i\ra \la j|),$$
where $\ds|\Om_A\ra:=\frac{1}{\sqrt{d_A}}\sum_{j=1}^{d_A} |jj\ra$ is called the \textit{maximally entangled vector} in $H_A\otimes H_A$. It is known that \cite{Cho75a, Sto82, HSR03, SSZ09}

\begin{itemize}
    \item $\Le$ is Hermitian preserving if and only if $C_{\Le}$ is Hermitian,

    \item $\Le$ is $k$-positive if and only if $C_{\Le}\in \BP_{k, AB}$, the set of \textit{$k$-block positive} operators (that is, $\la \xi|C_{\Le}|\xi\ra\geq 0$ for all $\xi\in H_A\otimes H_B$ such that ${\rm SR}(|\xi\ra )\leq k$),

    \item $\Le$ is CP if and only if $C_{\Le}\in \pos_{AB}$,

    \item $\Le$ is $k$-superpositive if and only if $C_{\Le}\in \sch_{k,AB}$,

    \item $\Le$ is EB if and only if $C_\Le \in \SEP_{AB}$.

    % \item $\Le$ is decomposable if and only if $C_{\Le}\in \dec_{AB}$, the set of \textit{decomposable} operators (that is, $C_{\Le}=X_1+(\id_A\otimes \top_B)(X_2)$ for some $X_1,X_2\in \pos_{AB}$),

    % \item $\Le$ is PPT if and only if $C_{\Le}\in \ppt_{AB}$, the set of positive operators on $H_A\otimes H_B$ which are of \textit{positive partial transpose (PPT)} (that is, both $C_{\Le}\in \pos_{AB}$ and $(\id_A\otimes \top_B)(C_{\Le})\in \pos_{AB}$ hold).
\end{itemize}

% Two further properties of CP maps related to quantum entanglement is PPT and entanglement-breaking (EB) property. A map $\Le$ is called EB if $C_{\Le}$ is in the set $\SEP_{AB}$ of separable operators. $\Le$ is called PPT if $C_{\Le}$ is in the set $\ppt_{AB}$ of PPT operators, i.e., both $C_{\Le}\in \pos_{AB}$ and $(\id_A\otimes \top_B)(C_{\Le})\in \pos_{AB}$ hold. Note that every EB map is PPT, but the converse does not hold in general (REF:Per96, Hor97). . 

% Other classes of $k$-superpositive (or $k$-partially entanglement breaking) maps has been considered as intermediate notion between CP and EB maps (CK06, SSZ09), $\Le$ is called {$k$-superpositive} if $C_{\Le}$ is in the set $\sch_{k,AB}$ of positive operators having Schmidt number $\leq k$. It is known (SSZ09) that $\Le$ is $k$-superpositive if and only if $\Le\in {\rm conv}\set{{\rm Ad}_K: K\in B(H_B,H_A), {\rm rank}(K)\leq k}$, where ${\rm Ad}_K(X):=KXK^*$ is a conjugation map. Moreover, $1$-superpositivity coincides with EB property, and $k$-superpositivity coincides with the complete positivity whenever $k\geq \min(d_A,d_B)$. 

We define the adjoint linear map $\Le^*\in B(B(H_B),B(H_A))$ with respect to the Hilbert-Schumidt inner product, i.e.,
    $${\rm Tr}(\Le(Z)^*W)={\rm Tr}(Z^* \Le^*(W)),\;\; Z,W\in B(H_A).$$
Recall that the adjoint operation $\Le\mapsto \Le^*$ preserves all the above properties, i.e. $k$-positivity, complete positivity, and $k$-superpositivity.

\subsection{Mapping cones and duality}

Let us briefly recall several notions in convex analysis and the theory of mapping cones. First, $B^{h}(B(H_A),B(H_B))$ is a real vector space equipped with an inner product
%--------- the same notation with SSZ09 and Sko11, but different from GKS21 and Kye23
\begin{equation} \label{eq-pairing1}
    \la \Phi,\Psi\ra:=\Tr(C_{\Phi}^*C_{\Psi})=\Tr(C_{\Phi}C_{\Psi}).
\end{equation}
For a subset $\K\subset B^h(B(H_A),B(H_B))$, the \textit{dual cone} $\K^{\circ}$ of $\K$ is defined by
\begin{equation}\label{eq-Horodecki}
    \K^{\circ}:=\set{\Phi\in B^{h}(B(H_A),B(H_B)):\la \Phi,\Psi\ra\geq 0\; \forall\, \Psi\in \K}. 
\end{equation}
It is well-known in convex analysis \cite{Roc70} that $\K^{\circ\circ}$ is the smallest closed convex cone containing $\K$. In particular, $\K$ is a closed convex cone if and only if $\K^{\circ\circ}=\K$. 
% Moreover, for two closed convex cones $\K_1,\K_2\subset B^h(B(H_A),B(H_B))$, we have
% \begin{equation}
%     (\K_1\vee \K_2)^{\circ}=\K_1^{\circ}\wedge \K_2^{\circ} \quad \text{and} \quad (\K_1\wedge \K_2)^{\circ}=\K_1^{\circ}\vee \K_2^{\circ},
% \end{equation}
% where $\K_1\vee\K_2:={\rm conv}(\K_1\cup \K_2)$ and $\K_1\wedge\K_2:=\K_1\cap \K_2$.

A closed convex cone $\K\subset \PO_{AB}$ is called a \textit{mapping cone} \cite{Sto86, Sko11} if it is invariant under the compositions by CP maps, i.e.,
\begin{equation}\label{eq10}
    \CP_{BB}\circ \K \circ \CP_{AA}\subset \K,
\end{equation}
where $\K_1\circ \K_2:=\set{\Phi\circ \Psi: \Phi\in \K_1,\; \Psi\in \K_2}$. Since the identity maps $\id_A,\id_B$ are CP maps, \eqref{eq10} is equivalent to $\CP_{BB}\circ \K \circ \CP_{AA}= \K$. 

% There are some important aspects on the study of mapping cones. First, 
If $\K$ is a nonzero mapping cone in $\PO_{AB}$, then the associated $K^{\circ}$ and $K^*:=\set{\Le^*:\Le\in \K}$ are mapping cones in $\PO_{AB}$ and $\PO_{BA}$, respectively \cite{Sko11, GKS21}.
% \;\; \top_B \circ \K,\;\; \K\circ \top_A, \;\; \K^*:=\set{\Le^*:\Le\in \K}, \;\; \K_1\vee \K_2, \;\; \K_1 \wedge \K_2.$$
% See for proofs. Second,
Moreover, all the classes
\begin{equation} \label{eq-posmaps}
    \PO, \PO_k, \CP, \SP_k, \EB
\end{equation}
introduced in Section \ref{sec-notations} are mapping cones. Note that there are many other mapping cones important in quantum information theory, such as the cone of PPT (Positive Partial Transpose) maps, decomposable maps, and {recently} $k$-entanglement breaking maps \cite{Sto82, CMHW19, DMS23}. Now, for a mapping cone $\K$, we can consider both the Choi correspondence $C_{\K}$ and the dual cone $\K^{\circ}$. Some important examples are exhibited in Table \ref{tab-mapping} below.

\begin{table}[h!]
  \begin{center}
    \caption{Mapping cones, Choi correspondences, and dual cones}
    \label{tab-mapping}
    \begin{tabular}{c||c|c|c|c|c} % <-- Alignments: 1st column left, 2nd middle and 3rd right, with vertical lines in between
       %$\backslash$ 
\hline
    $\K$ & $\PO$ & $\PO_{k}$ & $\CP$ & $\SP_k$ & $\EB$ \\
\hline
    $C_\K$  & $\BP$ & $\BP_k$ & $\pos$ & $\sch_k$ & $\SEP$\\
\hline
    $\K^{\circ}$ & $\EB$ & $\SP_k$ & $\CP$ & $\PO_k$ & $\PO$\\
\hline
    \end{tabular}
  \end{center}
\end{table}

% Temporary \textbf{REMARK}: The linear map in $\SP_k$ is called \textit{$k$-superpositive} or \textit{$k$-partially entanglement-breaking}, introduced in SZZ09 and CK06 respectively. These two names come from the fact
% \begin{align*}
%     \Phi\in \SP_{k,AB} & \iff \Phi=\sum_i {{\rm Ad}_{V_i}} \text{ for $V_i:H_A\to H_B$ with ${\rm rank}(V_i)\leq k$}\\
%     &\iff (\id_C\otimes \Phi)(\pos_{CA})\subset \sch_{k,CB} \; \text{ for every system $C$}.
% \end{align*}

% Finally and most importantly, mapping cones have special characterization theorems via duality (REF:Sko11, Sto12, GKS21, Kye23). More precisely, there are stronger relations between a pair of dual cones with a mapping
% cone symmetry than the relation between a mere pair of dual convex cones. For our purpose, we list a result revealing that the structure of mapping cones can be understood through \textit{ampliation} (REF: GKS21).

%Finally, and of utmost importance, mapping cones have been characterized by special theorems using \textit{duality}. Specifically, there exist stronger relationships between a pair of dual cones with mapping cone symmetry compared to the relationship between a simple pair of dual convex cones. For our purpose, we present a result that highlights the understanding of the structure of mapping cones through \textit{ampliation} \cite{GKS21}. A pairing between Hermitian matrices $X\in B^{h}(H_{AB})$ and Hermitian-preserving linear maps $\Le\in B^{h}(B(H_A),B(H_B))$ is defined by

Let us explain more direct connections between the Choi correspondences $C_{\K}$ and the dual cones $\K^{\circ}$. First of all, a natural pairing between Hermitian operators $X\in B^{h}(H_{AB})$ and Hermitian-preserving linear maps $\Le\in B^{h}(B(H_A),B(H_B))$ is given by
\begin{equation} \label{eq-pairing2}
    \la X,\Le\ra:=\Tr(C_{\Le}X)=\la \Om_A|(\id_A\otimes \Le^*)(X)|\Om_A\ra.
\end{equation}

% Note that two parings \eqref{eq-pairing2} and \eqref{eq-pairing1} has a direct relation $\la \Phi,\Psi\ra=\la C_{\Phi}, \Psi\ra$ for $\Phi,\Psi\in B^h(B(H_A),B(H_B))$. Therefore, the duality between $\K$ and $\K^{\circ}$ implies that
% \begin{itemize}
%     \item $\Le\in \K$ if and only if $\la X,\Le\ra\geq 0$ for $X\in C_{\K^{\circ}}$,

%     \item $X\in C_{\K}$ if and only if $\la X,\Le\ra\geq 0$ for $\Le\in \K^{\circ}$,    
% \end{itemize}
% for every closed convex cone $\K\subset B^h(B(H_A), B(H_B))$. However, we can give a further perspective when $\K$ is a mapping cone which has been noted in (REF:) 

%A famous theorem in this direction is the so-called {\it Horodecki} criterion \cite{HHH96}, which implies that $C_{\K}=\SEP$ and $\K^{\circ}=\PO$ determines each other. Furthermore, this perspective extends to the general mapping cones $\K\subseteq \PO_{A, B}$ as follows.

Then an extended form of {the famous {\it Horodecki criterion}} is given as follows with respect to the pairing (\ref{eq-pairing2}) above.

%The following proposition is an extension of the famous {\it Horodecki criterion} to general mapping cones $\K\subseteq \PO_{A,B}$. 

%Finally, and of utmost importance, mapping cones have been characterized by special theorems using \textit{duality}. Specifically, there exist stronger relationships between a pair of dual cones with mapping cone symmetry compared to the relationship between a simple pair of dual convex cones. For our purpose, we present a result that highlights the understanding of the structure of mapping cones through \textit{ampliation} \cite{GKS21}. A pairing between Hermitian matrices $X\in B^{h}(H_{AB})$ and Hermitian-preserving linear maps $\Le\in B^{h}(B(H_A),B(H_B))$ is defined by
%\begin{equation} \label{eq-pairing2}
 %   \la X,\Le\ra:=\Tr(C_{\Le}X)=\la \Om_A|(\id_A\otimes \Le^*)(X)|\Om_A\ra.
%\end{equation}
% (\textbf{REMARK}: I changed the defintion a little bit) Then the two pairings given by \eqref{eq-pairing1} and \eqref{eq-pairing2} can be linked with twirling operations as in the following Lemma 

\begin{proposition}{\cite[Proposition 4.1]{GKS21}} \label{prop-ampliation}
Suppose that a nonzero closed convex cone $\mathcal{K}\subset B^h(B(H_A),B(H_B))$ satisfies $\K\circ \CP_{AA}\subset \K$. 
\begin{enumerate}
    \item The following are equivalent for $\Le\in B(B(H_A),B(H_B))$:
\begin{enumerate}
    \item $\Le\in \mathcal{K}$,

    % \item $(\id_A\otimes \Le)(\pos_{AA})\subset C_{\K}$ (\textbf{needed?}),

    % \item $(\id_B\otimes \Le)(C_{(\K^{\circ})^*})\subset \pos_{BB}$ (\textbf{needed?}),

    \item $(\id_A\otimes \Le^*)(X)\in \pos_{AA}$ for every $X\in C_{\K^{\circ}}$.

    \item $\la X,\Le\ra\geq 0$ for every $X\in C_{\K^{\circ}}$.
\end{enumerate}
    \item The following are equivalent for $X\in B(H_A\otimes H_B)$:
\begin{enumerate}
    \item $X\in C_{\K}$,

    \item $(\id_A\otimes \Le^*)(X)\in \pos_{AA}$ for every $\Le\in \K^{\circ}$,

    %\item $(\Psi\otimes \id_B)(X)\in \pos_{BB}$ for every $\Psi\in \K^{\circ}$,

    \item $\la X,\Le\ra\geq 0$ for every $\Le\in \K^{\circ}$.

    %\item $\la \Om_B|(\Psi\otimes \id_B)(X)|\Om_B\ra\geq 0$ for every $\Psi\in \K^{\circ}$.
\end{enumerate}
\end{enumerate}

\end{proposition}

%This explains more detailed structures of mapping cones through \textit{ampliation} \cite{GKS21}, and we also refer to \cite{Sko11, Sto12} for other characterization results, namely via \textit{composition} and \textit{tensor product}. 

Note that Proposition \ref{prop-ampliation} can be applied for arbitrary mapping cone $\K$. For example, the {Horodecki criterion} is for a special case $C_{\K}=\SEP$ and $\K^{\circ}=\PO$, and it was used to study separability of quantum states \cite{HHH96}. Furthermore, Proposition \ref{prop-ampliation} has been applied to study quantum states with upper bounds on the Schmidt numbers \cite{TH00}, decomposable maps \cite{Sto82}, $k$-positive maps \cite{EK00, TH00}, and $k$-superpositive maps \cite{SSZ09}. 

Very recently, \cite{PJPY23} proposed an optimized form of Proposition \ref{prop-ampliation} under compact group symmetries for a special case where $C_{\K}=\PO$ and $\K^{\circ}=\SEP$. Furthermore, they introduced a concrete application to analyze separability of quantum states. One of the main contributions of this paper is to extend their results for general mapping cones $\K$ under compact group symmetries. See Section \ref{sec-mapping_cones} and Corollary \ref{cor-extreme} for more details.

\vspace{10mm}

% \begin{remark}
% REF: GKS21, Kye23
% All the properties listed in this section holds if we even weaken the definition of mapping cone $\K$ by removing the condition $\K\subset \PO_{AB}$.. But according to ...
% % Our definitions of $\K^{\circ}$ and $\Le^*$ are a bit different from \cite{GKS21,Kye23}. Nevertheless, the above results remain to hold after a little modification of the proof, using two relations (up to multiplicative constant) 
% % \begin{align*}
% %     \la \Phi,\Psi\ra&=\la \Psi^*,\Phi^*\ra,\\
% %     \la \Phi\circ \Psi,\Le\ra&=\la \Psi,\Phi^*\circ \Le\ra = \la \Phi,\Le\circ\Psi^*\ra.
% % \end{align*}
% \end{remark}

\subsection{Compact group symmetries and twirling operations}
%We first recall the notions regarding group symmetry from \cite{PJPY23}.

Let $\pi:G\to \U(H)$, $\pi_A:G\to \U(H_A)$ and $\pi_B:G\to \U(H_B)$ be (finite-dimensional) {\it unitary representations} of a compact Hausdorff group $G$. An operator $X\in B(H)$ is called {\it $\pi$-invariant} if 
\begin{equation}
\pi(x)X= X \pi(x)    
\end{equation}
for all $x\in G$, and a linear map $\Le:B(H_A)\rightarrow B(H_B)$ is called {\it $(\pi_A,\pi_B)$-covariant} if 
\begin{equation}
    \Le(\pi_A(x)Z\pi_A(x)^*)=\pi_B(x)\Le(Z)\pi_B(x)^*
\end{equation} 
for all $x\in G$ and $Z\in B(H_A)$. We denote by ${\rm Inv}(\pi)$ the set of all $\pi$-invariant operators in $B(H)$, and by ${\rm Cov}(\pi_A,\pi_B)$ the set of all $(\pi_A,\pi_B)$-covariant maps in $B(B(H_A),B(H_B))$. A unitary representation $\pi$ is called \textit{irreducible} if ${\rm Inv}(\pi)=\Comp \cdot I_d$. If $\pi$ is irreducible, so is the \textit{contragredient representation} $\overline{\pi}:G\to \U(H)$. Here, $\overline{\pi}(x)=\overline{\pi(x)}$ for $x\in G$. For unitary representations $\pi_A:G\rightarrow \U(H_A)$ and $\pi_B:G\rightarrow \U(H_B)$, the {\it tensor representation} $\pi_A\otimes \pi_B:G\rightarrow \U(H_A\otimes H_B)$ is given by $(\pi_A\otimes \pi_B)(x)=\pi_A(x)\otimes \pi_B(x)$ for $x\in G$.

We are interested in two types of standard averaging techniques. The \textit{$\pi$-twirling} operation on $B(H)$ is defined by
\begin{equation}
    \mathcal{T}_{\pi}X:=\int_G[{\rm Ad}_{\pi(x)}(X) ]dx=\int_{G}[\pi(x)X\pi(x)^*]dx
\end{equation}
for any operator $X\in B(H)$ and a unitary representation $\pi$ of $G$. Moreover, the {\it $(\pi_A,\pi_B)$-twirling} operation on $B(B(H_A),B(H_B))$ is defined by
\begin{align}
    \mathcal{T}_{\pi_A,\pi_B}\Phi:=\int_G[{\rm Ad}_{\pi_B(x^{-1})}\circ \Phi\circ {\rm Ad}_{\pi_A(x)}]dx
\end{align}
for any linear maps $\Phi\in B(B(H_A),B(H_B))$ and unitary representations $\pi_A,\pi_B$ of $G$. Here, $dx$ is the (normalized) Haar measure of $G$.

Let us collect some useful properties of the twirling operations for later uses. First of all, $\mathcal{T}_{\pi}$ is a \textit{conditional expectation} onto the $*$-subalgebra ${\rm Inv}(\pi)$ of $B(H)$ in the sense that $\Tt_{\pi}\circ \Tt_{\pi}=\Tt_{\pi}$ and the range of $\Tt_{\pi}$ is $\Inv(\pi)$. Similarly, $\Tt_{\pi_A,\pi_B}$ can be seen as a projection onto the space $\Cov(\pi_A,\pi_B)$. These two operations satisfy the following properties.

\begin{proposition} \cite{PJPY23} \label{prop-twirling}
Let $\pi:G\to \U(H)$, $\pi_A:G\to \U(H_A)$ and $\pi_B:G\to \U(H_B)$ be unitary representations of $G$. Then we have the following.
\begin{enumerate}
    \item ${\rm Tr}((\mathcal{T}_{\pi}Z)^*\,W)={\rm Tr}(Z^* (\mathcal{T}_{\pi}W))$ for any $Z,W\in B(H)$.
    
    % \item $\mathcal{T}_{\pi_A\otimes \pi_B}\circ (\top\otimes {\rm id}_B)=(\top\otimes {\rm id}_B)\circ \mathcal{T}_{\overline{\pi_A}\otimes \pi_B}$.
    
    \item $\left (\mathcal{T}_{\pi_A,\pi_B}\Phi\right )^*=\mathcal{T}_{\pi_B,\pi_A}(\Phi^*)$ for any $\Phi\in B(B(H_A), B(H_B))$.
    
    \item $C_{(\mathcal{T}_{\pi_A,\pi_B}\Phi)}=\mathcal{T}_{\overline{\pi_A}\otimes {\pi_B}}\left ( C_{\Phi}\right )$ for any $\Phi \in B(B(H_A), B(H_B))$.
\end{enumerate}
\end{proposition}

Since invariant operators and covariant linear maps are the images of the twirling projections, Proposition \ref{prop-twirling} (2), (3), and (4) imply the following conclusions.

\begin{corollary} \cite{PJPY23}\label{cor-twirling}
Let $X\in B(H_A\otimes H_B)$ be a bipartite operator and $\Phi:B(H_A)\to B(H_B)$ be a linear map. Then
\begin{enumerate}
    % \item $X\in {\rm Inv}(\pi_A\otimes \pi_B)$ if and only if $(\top \otimes {\rm id})(X)\in {\rm Inv}(\overline{\pi_A}\otimes \pi_B)$. 
    
    \item $\Phi\in {\rm Cov}(\pi_A,\pi_B)$ if and only if $\Phi^*\in {\rm Cov}(\pi_B,\pi_A)$.
    
    \item $\Phi\in {\rm Cov}(\pi_A,\pi_B)$ if and only if $C_{\Phi}\in {\rm Inv}(\overline{\pi_A}\otimes \pi_B)$.
\end{enumerate}
\end{corollary}

% The following properties are straightforward.

% \begin{proposition}
% Let $\pi_A,\pi_B,\pi_C$ be representations of $G$.
% \begin{enumerate}
%     \item If $\Phi\in {\rm Cov}(\pi_A,\pi_B)$ and $\Psi\in {\rm Cov}(\pi_B,\pi_C)$, then $\Psi\circ \Phi\in {\rm Cov}(\pi_A,\pi_C)$.

%     \item If $X\in {\rm Inv}(\pi_A\otimes \pi_B)$ and $\Phi\in {\rm Cov}(\pi_B,\pi_C)$, then $(\id_A\otimes \Phi)(X)\in {\rm Inv}(\pi_A\otimes \pi_C)$.

%     \item $\Phi\in {\rm Cov}(\pi_A,\pi_B)$ if and only if $(\id_A\otimes \Phi)(X)\in {\rm Inv}(\overline{\pi_A}\otimes \pi_B)$ for every $X\in {\rm Inv}(\overline{\pi_A}\otimes \pi_A)$.

%     \item If $\Phi\in {\rm Cov}(\pi_A^G,\pi_B^G)$ and $\Psi\in {\rm Cov}(\pi_{A'}^H,\pi_{B'}^H)$, then $\Phi\otimes \Psi\in {\rm Cov}(\pi_{AA'}^{G\times H}, \pi_{BB'}^{G\times H})$, where $\pi_{AA'}^{G\times H}(g,h):=\pi_A(g)\otimes \pi_{A'}(h)$.
% \end{enumerate}
% \end{proposition}
% \begin{proof}
% For the sufficiency in $(3)$, we may choose $X=|\Om_A\ra\la \Om_A|\in {\rm Inv}(\overline{\pi_A}\otimes \pi_A)$. Then $(\id_A\otimes \Phi)(X)=C_{\Phi}\in {\Inv}(\overline{\pi_A}\otimes \pi_B)$ if and only if $\Phi\in \Cov(\pi_A,\pi_B)$ by Proposition \ref{prop-twirling}.
% \end{proof}

\vspace{5mm}

\section{Mapping cones with compact group symmetry}\label{sec-mapping_cones}

From now on, we describe how standard duality results between mapping cones can be naturally carried over into our framework of compact group symmetry. Such a strategy of optimizing the duality relations under compact group symmetries was pursued in \cite{PJPY23} for a special mapping cone $\K=\EB$, where we have $C_{\K}=\SEP$ and $\K^{\circ}=\PO$. The authors applied the optimized Horodecki criterion to study separability of invariant quantum states under the standard symmetries of the signed symmetric group (or the hyperoctahedral group). In this section, we prove that their approach covers not only for the special $\K=\EB$, but also for general mapping cones $\K$ under compact group symmetries. Then, in Section \ref{sec-application},  we apply the generalized results to characterize $k$-positivity of all orthogonally invariant quantum states.

\begin{lemma} \label{lem-twirlparing}
Let $\pi_A,\pi_B$ be two unitary representations of $G$. 
\begin{enumerate}
    \item For $\Phi,\Psi \in B(B(H_A),B(H_B))$, we have
\begin{equation} \label{eq-twirlpair1}
    \la \Tt_{\pi_A,\pi_B}\Phi,\Psi\ra=\la \Phi,\Tt_{\pi_A,\pi_B}\Psi\ra.
\end{equation}
\item For $X\in B(H_A\otimes H_B)$ and $\Le\in B(B(H_A), B(H_B))$, we have
\begin{equation} \label{eq-twirlpair2}
    \la \mathcal{T}_{\overline{\pi_A}\otimes \pi_B}X, \Le \ra=\la X, \mathcal{T}_{\pi_A,\pi_B}\Le\ra.
\end{equation}
\end{enumerate}
\end{lemma}
\begin{proof}
Both two identities follow from Proposition \ref{prop-twirling}. Indeed, we have
\begin{align*}
    \la \Tt_{\pi_A,\pi_B}\Phi, \Psi\ra &= \Tr((C_{(\Tt_{\pi_A,\pi_B}\Phi)})^*C_{\Psi})= \Tr((\Tt_{\overline{\pi_A}\otimes \pi_B}C_\Phi)^*C_{\Psi})\\
    &=\Tr(C_\Phi^* \Tt_{\overline{\pi_A}\otimes \pi_B}C_{\Psi}) = \Tr(C_{\Phi}^*C_{(\Tt_{\pi_A,\pi_B}\Psi)})=\la \Phi, \Tt_{\pi_A,\pi_B}\Psi \ra
\end{align*}
which implies \eqref{eq-twirlpair1}. Also, we can repeat a similar argument to prove \eqref{eq-twirlpair2}.
\end{proof}

Recall that positivity, complete positivity, and EB property are preserved under the twirling operation $\Le\mapsto \Tt_{\pi_A,\pi_B}\Le$ \cite[Proposition 2.1]{PJPY23}.  More generally, we have $\Tt_{\pi_A,\pi_B}(\K) \subseteq \K$ for any mapping cone $\K$ by the following Proposition \ref{prop-twirlmapping}.

\newpage

\begin{proposition} \label{prop-twirlmapping}
For a closed convex cone $\K\in B^h(B(H_A),B(H_B))$, the following are equivalent:
\begin{enumerate}
    \item $\Tt_{\pi_A,\pi_B}\K\subset \K$,

    \item $\Tt_{\pi_A,\pi_B}(\K^{\circ})\subset \K^{\circ}$,

    \item $\Tt_{\pi_B,\pi_A}(\K^{*})\subset \K^{*}$,
    
    \item $\Tt_{\overline{\pi_A}\otimes \pi_B}C_{\K}\subset C_{\K}$.
\end{enumerate}
In this case, we have $\Tt_{\overline{\pi_A}\otimes \pi_B}C_{\K}={\Inv}(\overline{\pi_A}\otimes \pi_B)\cap C_{\K}$ and $\Tt_{\pi_A,\pi_B}\K={\rm Cov}(\pi_A,\pi_B)\cap \K$. Moreover, the above conditions hold if $\CP_{BB}\circ \K\circ \CP_{AA}\subset \K$.
\end{proposition}
\begin{proof}
    The equivalence $(1)\Leftrightarrow (3) \Leftrightarrow (4)$ is a direct consequence from Proposition \ref{prop-twirling}. For $(1)\Rightarrow (2)$, it is enough to see that 
        $$\la \Tt_{\pi_A,\pi_B}\Phi,\Psi\ra=\la \Phi,\Tt_{\pi_A,\pi_B}\Psi\ra\geq 0$$
for any $\Phi\in \K^{\circ}$ and $\Psi\in \K$. The other direction $(2)\Rightarrow (1)$ follows from $(1)\Rightarrow (2)$ since $\K^{\circ\circ}=\K$.
    
The second conclusion follows from the properties $\Tt_{\overline{\pi_A}\otimes \pi_B}\circ \Tt_{\overline{\pi_A}\otimes \pi_B}=\Tt_{\overline{\pi_A}\otimes \pi_B}$ and $\Tt_{\pi_A,\pi_B}\circ \Tt_{\pi_A,\pi_B}=\Tt_{\pi_A,\pi_B}$. For the last assertion, it is enough to note that the twirling operations preserve positivity of linear maps, and that $\Tt_{\pi_A,\pi_B}\Phi$ is approximated by convex combinations of ${\rm Ad}_{\pi_B(x)^*}\circ \Phi\circ {\rm Ad}_{\pi_A(x)}\in \K$ for $x\in G$.
\end{proof}

Recall that $\K$ and $C_{\K^{\circ}}$ determine each other via the generalized Horodecki criterion (Proposition \ref{prop-ampliation}). One of our main results in this section is to establish an analogous result for ${\rm Cov}(\pi_A,\pi_B)\cap \K$ and ${\Inv}(\overline{\pi_A}\otimes \pi_B)\cap C_{\K^{\circ}}$ as follows.

\begin{theorem} \label{thm-main1}
Let $\K\subset B^h(B(H_A),B(H_B))$ be a closed convex cone satisfying $\CP_{BB}\circ \K\circ \CP_{AA}\subset \K$. 
\begin{enumerate}
    \item The following are equivalent for $\Le\in {\rm}Cov(\pi_A,\pi_B)$:
\begin{enumerate}
    \item $\Le\in {\rm Cov}(\pi_A,\pi_B)\cap \K$,

    % \item $\Le^*\circ \Phi\in \CP$ for every $\Le\in {\rm Cov}(\pi_A,\pi_B)\cap \K^{\circ}$,

    \item $(\id_A\otimes \Le^*)(X)\in \pos_{AA}$ for every $X\in {\Inv}(\overline{\pi_A}\otimes \pi_B)\cap C_{\K^{\circ}}$,

    \item $\la X, \Le\ra\geq 0$ for every $X\in {\Inv}(\overline{\pi_A}\otimes \pi_B)\cap C_{\K^{\circ}}$.
\end{enumerate}
\item The following are equivalent for $X\in {\rm Inv}(\overline{\pi_A}\otimes \pi_B)$:
\begin{enumerate}
    \item $X\in {\Inv}(\overline{\pi_A}\otimes \pi_B)\cap C_{\K}$,

    \item $(\id_A\otimes \Le^*)(X)\in \pos_{AA}$ for every $\Le \in {\rm Cov}(\pi_A,\pi_B)\cap \K^{\circ}$,

    \item $\la X,\Le\ra\geq 0$ for every $\Le \in {\rm Cov}(\pi_A,\pi_B)\cap \K^{\circ}$.
\end{enumerate}
\end{enumerate}
\end{theorem}

\begin{proof}
Let us prove only the first part. The other one is analogous. Note that $(a)\Rightarrow (b)$ follows from Proposition \ref{prop-ampliation} and $(b)\Rightarrow (c)$ is clear from the relation \eqref{eq-pairing2}, so it suffices to prove the direction $(c)\Rightarrow (a)$. Since $\Le\in \Cov(\pi_A, \pi_B)$, we have 
    $$\la X, \Le\ra = \la X, \Tt_{\pi_A,\pi_B}\Le\ra = \la \Tt_{\overline{\pi_A}\otimes \pi_B}X, \Le\ra$$
for all $X\in C_{\K^{\circ}}$ by Lemma \ref{lem-twirlparing}. Now $\Tt_{\overline{\pi_A}\otimes \pi_B}X\in {\rm Inv}(\overline{\pi_A}\otimes \pi_B)\cap C_{\K^{\circ}}$ by Proposition \ref{prop-twirlmapping}, so the assumption $(c)$ implies that $\la X,\Le\ra\geq 0$ for all $X\in C_{\K^{\circ}}$. Therefore, Proposition \ref{prop-ampliation} implies $\Le3 \in \K$ again.
\end{proof}

% \begin{proof}
% Let us prove only the first assertion. The other one is analogous. Note that $(1)\Rightarrow (2)$ follows from Proposition \ref{prop-ampliation} and $(2)\Rightarrow (3)$ is clear from the relation \eqref{eq-pairing2}, so it suffices to prove the direction $(3)\Rightarrow (1)$. Since $\Le\in \Cov(\pi_A, \pi_B)$, we have 
%     $$\la X, \Le\ra = \la X, \Tt_{\pi_A,\pi_B}\Le\ra = \la \Tt_{\overline{\pi_A}\otimes \pi_B}X, \Le\ra$$
% for all $X\in C_{\K^{\circ}}$ by Lemma \ref{lem-twirlparing}. Now $\Tt_{\overline{\pi_A}\otimes \pi_B}X\in {\rm Inv}(\overline{\pi_A}\otimes \pi_B)\cap C_{\K^{\circ}}$ by Proposition \ref{prop-twirlmapping}, so the assumption $(3)$ implies that $\la X,\Le\ra\geq 0$ for all $X\in C_{\K^{\circ}}$. Therefore, Proposition \ref{prop-ampliation} again implies that $\Le\in \K$.

% \end{proof}

Note that ${\rm Cov}(\pi_A,\pi_B)\cap \K$ plays as witnesses for ${\Inv}(\overline{\pi_A}\otimes \pi_B)\cap C_{\K^{\circ}}$ when $\CP_{BB}\circ \K\circ \CP_{AA}\subset \K$. We can further prove that much fewer witnesses from ${\rm Cov}(\pi_A,\pi_B)\cap \K$ are enough if $\K$ is a nonzero mapping cone, i.e., $\K\subset \PO_{AB}$. Let us start with a compact convex subset
    $$\Covn(\pi_A,\pi_B)\cap \K:=\set{\Phi\in {\Cov}(\pi_A,\pi_B)\cap \K:\Tr C_{\Phi}=1}.$$
% Note that $\cov(\pi_A,\pi_B)$ contains two other subsets
% \begin{align*}
%     {\covU}(\pi_A,\pi_B)&:=\set{\Phi\in {\Cov}(\pi_A,\pi_B): \Phi(I_A/d_A)=I_B/d_B},\\
%     &=\set{\Phi\in {\Cov}(\pi_A,\pi_B):(\Tr\otimes \id_B)(C_{\Phi})=I_B/d_B}\\
%     {\CovTP}(\pi_A,\pi_B)&:=\set{\Phi\in {\Cov}(\pi_A,\pi_B): \Phi \text{ is TP}}\\
%     &=\set{\Phi\in {\Cov}(\pi_A,\pi_B):(\id_A\otimes \Tr)(C_{\Phi})=I_A/d_A}.
% \end{align*}
% whose inclusions are in general strict. However, 
 %However, the following lemma claims that they become the same (up to multiplicative constant) under some mild conditions of $\pi_A,\pi_B$.

\begin{remark}\label{lem-covIrred} 
\cite[Lemma 3.6]{PJPY23} implies that $(d_B/d_A)\cdot \Covn(\pi_A,\pi_B)$ is the set of all unital $(\pi_A,\pi_B)$-covariant maps if $\pi_A$ is irreducible, and that $\Covn(\pi_A,\pi_B)$ is the set of all trace-preserving $(\pi_A,\pi_B)$-covariant maps if $\pi_B$ is irreducible. Note that the notation $\Covn(\pi_A,\pi_B)$ is used differently in \cite{PJPY23}.
\end{remark}

Now we can extend Theorem 3.9 (6) of \cite{PJPY23} to a general context of mapping cones. More precisely, Theorem \ref{thm-main1} implies that only the extreme points of $\Covn(\pi_A,\pi_B)\cap \K^{\circ}$ are enough as witnesses for $\Inv(\overline{\pi_A}\otimes \pi_B)\cap C_\K$ since every compact convex set in a finite-dimensional space can be written as the convex hull of of its extreme points.

\begin{corollary} \label{cor-extreme}
If $\mathcal{K}\subset \PO_{AB}$ is a nonzero mapping cone, then the following are equivalent for $X\in \Inv(\overline{\pi_A}\otimes \pi_B)$:
\begin{enumerate}
    \item $X\in \Inv(\overline{\pi_A}\otimes \pi_B)\cap C_\K$,

    \item $(\id_A\otimes \Le^*)(X)\in \pos_{AA}$ for every $\Le\in {\rm Ext}(\Covn(\pi_A,\pi_B)\cap \K^{\circ})$,

    \item $\la X,\Le\ra\geq 0$ for every $\Le\in {\rm Ext}(\Covn(\pi_A,\pi_B)\cap \K^{\circ})$.
\end{enumerate}
% If $\pi_B$ is irreducible, then we still have the same equivalences by replacing ${\rm CovTP}$ with ${\rm Cov}_1$.
\end{corollary}

Recall that the main purpose of Theorem 3.9 (6) in \cite{PJPY23} was to use ${\rm Ext}(\Covn(\pi_A,\pi_B)\cap \PO)$ to study separability of $\rho\in \Inv(\overline{\pi_A}\otimes \pi_B)\cap \D$. A merit of our general approach is that Corollary \ref{cor-extreme} with the case $\K=\SP_k$ provides a new systematic way to compute the Schmidt numbers of $\rho\in \Inv(\overline{\pi_A}\otimes \pi_B)\cap \D$ using ${\rm Ext}(\Covn(\pi_A,\pi_B)\cap \PO_k)$ as a family of witnesses.

\begin{theorem} \label{cor-invSch}
Let $\rho\in {\rm Inv}(\overline{\pi_A}\otimes \pi_B)$. Then the following are equivalent:
\begin{enumerate}
    \item ${\rm SN}(\rho)\leq k$,

    \item $(\id_A\otimes \Le)(\rho)\in \pos_{AA}$ for every $\Le\in {\rm Ext}(\Covn(\pi_B,\pi_A)\cap \PO_k)$,

    \item $\la \Om_A|(\id_A\otimes \Le)(\rho)|\Om_A\ra \geq 0$ for every $\Le\in {\rm Ext}(\Covn(\pi_B,\pi_A)\cap \PO_k)$.
\end{enumerate}
\end{theorem}

The above Theorem \ref{cor-invSch} will be applied for concrete applications in Section \ref{sec-application}.

%We remark that Theorem \ref{cor-invSch} together with Remark \ref{lem-covIrred} allows us to recover \cite[Theorem 3.9]{PJPY23} when $k=1$. Moreover, this provide a new method to find efficient Schmidt number witness.

% \begin{corollary} \label{cor-cov}
% For a linear map $\Le\in {\rm Cov}(\pi_B\otimes \pi_A)$, TFAE:
% \begin{enumerate}
%     \item $\Le$ is $k$-positive / $k$-superpositive / PPT / decomposable,
    
%     \item $(id_A\otimes \Le)\geq 0$ for every $X\in {\rm Inv}(\overline{\pi_A},\pi_B)$ which is positive with Schmidt number $\leq k$ / $k$-block positive / decomposable / PPT,

%     \item $\la X, \Le\ra\geq 0$ for every $X\in {\rm Inv}(\overline{\pi_A},\pi_B)$ which is positive with Schmidt number $\leq k$ / $k$-block positive / decomposable / PPT,
% \end{enumerate}
% \end{corollary}

%-------------------------------

\section{Quantum objects with orthogonal group symmetry}\label{sec-application}

ddThere are very few examples whose $k$-positivity or Schmidt numbers have been fully characterized. Even in the following cases
\begin{align}
    &\Le^{(d)}_{a,b}(Z):= (1-a-b)\frac{\Tr(Z)}{d}I_d+aZ+bZ^{\top}, \label{eq-OOCov}\\ 
    &\rho_{a,b}^{(d)}:=\frac{1-a-b}{d^2}I_d\otimes I_d+a|\Om_d\ra\la \Om_d|+\frac{b}{d}{F_d}, \label{eq-OOInv}
\end{align}
their $k$-positivity and Schmidt numbers have not been fully characterized. Here, $\ds |\Om_d\ra:=\frac{1}{\sqrt{d}}\sum_{j=1}^d|jj\ra$ is the maximally entangled state and $\ds F_d:=\sum_{i,j=1}^d |ij\ra\la ji|$ is the flip matrix. Let us write $\Le_{a,b}:=\Le^{(d)}_{a,b}$ and $\rho_{a,b}:=\rho^{(d)}_{a,b}$ for simplicity. On the problem of $k$-positivity, the answers for some special cases $\Le_{a,0}$, $\Le_{0,b}$, and $\Le_{a,1-a}$, were obtained from \cite{Tom85}, and some other cases were considered in \cite{TT83}. On the other hand, Schmidt numbers of the \textit{isotropic states} $\rho_{a,0}$ \cite{HH99} were computed in \cite{TH00}, and Schmidt numbers of the \textit{Werner states} $\rho_{0,b}$ \cite{Wer89} were also known (see \cite[Theorem 1.7.4]{Kye23lect} for example). Despite the partial answers to the cases of single parameters, the problems of the general cases $\Le_{a,b}$ and $\rho_{a,b}$ remain open. To our best knowledge, our result is the first example of computations of the Schmidt numbers for non-trivial two-dimensional families of quantum states in arbitrarily high dimensional settings.

A crucial observation is that the above quantum objects $\Le_{a,b}$ and $\rho_{a,b}$ are linked via the standard orthogonal group symmetries. Let $G$ be the orthogonal group $\mathcal{O}(d)$, and let $\pi_A(O)=\pi_B(O)=O$ be the standard representation of $\mathcal{O}(d)$. In this case, we denote by $\Cov(O,O)=\Cov(\pi_A,\pi_B)$ and $\Inv(O\otimes O)=\Inv(\overline{\pi_A}\otimes \pi_B)$ for simplicity. Then we have $\rho_{a,b}=C_{\Le_{a,b}}$ and $\Covn(O,O)=\set{\Le_{a,b}:a,b\in \Comp}$, as noted in \cite{VW01, Has18}. Moreover, $\Le_{a,b}$ is Hermitian-preserving if and only if $a,b\in \Real$.

In this section, we aim to establish a complete characterization of $k$-positivity of $\Le_{a,b}$ and Schmidt number of $\rho_{a,b}$ in terms of the parameters $a$ and $b$. Then $k$-block positivity of $\rho_{a,b}$ and $k$-superpositivity of $\Le_{a,b}$ are immediate through the Choi-Jamio{\l}kowski map.

Our strategy consists of two steps. The first step is to employ some methodologies from \cite{Tom85} to study $k$-positivity of $\Le_{a,b}\in \Covn(O,O)$ in a direct way (Theorem \ref{thm-OOkpos}), and the second step is to apply Theorem \ref{cor-invSch} to compute the Schmidt numbers of all $\rho_{a,b}\in {\rm Inv}(O\otimes O)\cap \D$ accurately (Theorem \ref{thm-OOSch}).

\subsection{$k$-positivity of orthogonally covariant maps}

Note that positivity and complete positivity of $\Le_{p,q}$ were completely characterized recently in a more general setting, namely the \textit{hyperoctahedrally covariant maps} \cite[Section 4]{PJPY23}. This section is devoted to characterizing $k$-positivity of all $\Le_{p,q}\in \Covn(O,O)$, which generalizes the results from \cite{Tom85}. Indeed, for the following convex subsets 
\begin{equation} \label{eq-kposset}
    P_k:=\set{(p,q)\in \Real^2:\Le_{p,q}\in \PO_k},~1\leq k\leq d,
\end{equation}
we prove $P_1\supsetneq P_2\supsetneq \cdots \supsetneq P_d$ with explicit geometric and algebraic descriptions. First of all, the geometric structures of the convex subsets $P_k$ can be categorized into four distinct cases.

\vspace{4pt}

\begin{enumerate}
    \item The region $P_1$ is \textit{trapezoid-shaped} with vertices $(1,0)$, $(0,-\frac{1}{d-1})$, $(-\frac{1}{d-1},0)$, and $(0,1)$.

\vspace{3pt}
    
    \item If $1<k\leq \frac{d}{2}$, the region $P_k$ is \textit{quadrilateral-shaped} with vertices $(1,0)$, $(0,-\frac{1}{d-1})$, $(-\frac{1}{kd-1},0)$, and $(-\frac{1}{kd+k-1}, \frac{k}{kd+k-1})$.

    \vspace{3pt}
    
    \item If $\frac{d}{2}<k<d$, the region $P_k$ is bounded by a piecewise-linear curve joining $(-\frac{2}{d^2+d-2},\frac{d}{d^2+d-2}), (1,0), (0,-\frac{1}{d-1})$, and $(-\frac{1}{kd-1},0)$ in that order, and by a conic (i.e. a quadratic curve) passing through $(-\frac{1}{kd-1},0)$ and $(-\frac{2}{d^2+d-2},\frac{d}{d^2+d-2})$.

    % {\color{blue}If $\frac{d}{2}<k<d$, the region $P_k$ is bounded by a piecewise-linear curve that connects the points $(-\frac{2}{d^2+d-2},\frac{d}{d^2+d-2}), (1,0)$, and $(-\frac{1}{kd-1},0)$ in that order. Additionally, it is bounded by a conic (a quadratic curve) that passes through the points $(-\frac{1}{kd-1},0)$ and $(-\frac{2}{d^2+d-2},\frac{d}{d^2+d-2})$.}

\vspace{3pt}

    \item Lastly, the region $P_d$ is a \textit{triangle} with vertices $(1,0), (0,-\frac{1}{d-1})$, and $(-\frac{2}{d^2+d-2}, \frac{d}{d^2+d-2})$.
\end{enumerate}

\vspace{3pt}

A visualization of the above characterizations for special cases $d=3$ and $d=4$ are given in the following Figure \ref{fig-OOkpos}.
% {\color{red} Note that for $d=4$, }

\begin{figure}[htb!] 
    \centering
    \includegraphics[scale=0.52]{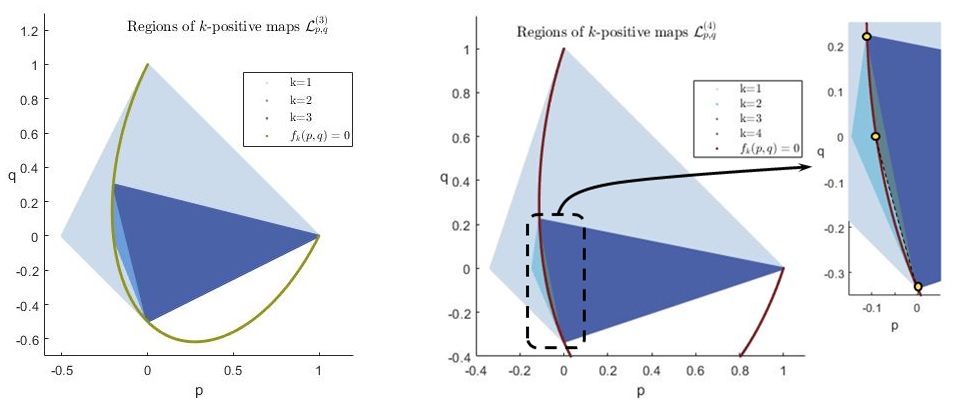}
    \caption{The regions of $k$-positive maps $\mathcal{L}_{p,q}^{(d)}$ for $d=3,4$}
    \label{fig-OOkpos}
\end{figure}

\newpage

We now present explicit algebraic descriptions of the regions $P_k$ (equivalently, $\PO_k$) in the following theorem.

\begin{theorem} \label{thm-OOkpos}
Let $\Le_{p,q}$ be a linear map of the form \eqref{eq-OOCov}. Then
\begin{enumerate}[leftmargin=*]
    \item $\Le_{p,q}\in \PO$ if and only if $\begin{cases} p-(d-1)q\leq 1, \\ q-(d-1)q\leq 1,\\-\frac{1}{d-1}\leq p+q\leq 1.\end{cases}$ %(if and only if $\Le_{p,q}\in \DEC$).
    
    \item $\Le_{p,q}\in \PO_k$ ($1<k\leq \frac{d}{2}$) if and only if $\begin{cases} p-(d-1)q\leq 1, \\ p+(d+1)q\leq 1, \\ (kd-1)p+(d-1)q\geq -1,\\ q-(kd-1)p\leq 1.\end{cases}$
    
    \item $\Le_{p,q}\in \PO_k$ ($\frac{d}{2}<k<d$) if and only if $\begin{cases} p-(d-1)q\leq 1, \\ p+(d+1)q\leq 1, \\ (kd-1)p+(d-1)q\geq -1,\\ f_k(x,y)\leq 0,\end{cases}$ where $f_k(x,y)$ is a quadratic polynomial explicitly given by
    \begin{align}
        f_k(x,y)=&(kd-1)x^2-(d^3-kd^2-kd-d+2)xy+(d-1)y^2 \nonumber \\
        &-(kd-2)x-(d-2)y-1. \label{eq-conic1}
    \end{align}
    
    \item $\Le_{p,q}\in \CP$ if and only if $\begin{cases} p-(d-1)q\leq 1, \\ p+(d+1)q\leq 1, \\ (d+1)p+q\geq -\frac{1}{d-1}.\end{cases}$ 
\end{enumerate}
\end{theorem}

% {\color{red}\textbf{ADD} remark: generalization of TT83, Tom85, and ...}

As mentioned already, (1) and (4) of Theorem \ref{thm-OOkpos}, i.e. positivity and complete positivity of $\Le_{p,q}$, were fully characterized in a recent paper \cite[Section 4]{PJPY23} for more general models. The authors studied linear maps of the form
\begin{equation}
    \psi_{a,b,c}(Z):=a\frac{\Tr(Z)}{d}I_d+bZ+cZ^{\top}+(1-a-b-c){\rm diag}(Z),
\end{equation}
and exhibited a full characterization of positivity and complete positivity. Thus, our main focus is to prove (2) and (3) of Theorem \ref{thm-OOkpos}. Let us recall a criterion of $k$-positivity proposed in \cite[Lemma 1]{Tom85}.

\begin{proposition} \label{prop-kpos}
Let $1\leq k\leq d$. Then a linear map $\Le:M_d\to M_d$ is $k$-positive if and only if the bipartite matrix
   \begin{equation}\label{eq40}
       C_k^v(\Le):=\sum_{i,j=1}^k |i\ra\la j|\otimes \Le(|v_i\ra\la v_j|)\in M_k\otimes M_d
   \end{equation}
is positive semidefinite for any choice of an orthonormal subset $\set{v_1,\ldots, v_k}$ of $\Comp^d$.
\end{proposition}

The following lemma plays a crucial role in applying the above Proposition \ref{prop-kpos} to prove (2) and (3) of Theorem \ref{thm-OOkpos}.

\begin{lemma} \label{lem-OOoptimization}
For $1\leq k\leq d$, we have
\begin{equation} \label{eq-optimize}
    \min \sum_{j,j'=1}^k|\la {v_j}|\overline{v_{j'}}\ra|^2=\max(2k-d,0),
\end{equation}
where $\overline{w}$ is the complex conjugation of $w\in \Comp^d$ and the minimum is taken over all orthonormal subsets $\set{v_1,\ldots, v_k}\subset \Comp^d$.
\end{lemma}
\begin{proof}
If $2k\leq d$, then we can take $|v_j\ra=\frac{1}{\sqrt{2}}(|2j-1\ra+i|2j\ra)$ ($j=1,\ldots, k$) and check $\la {v_j}|\overline{v_{j'}}\ra=0$ for every $j,j'$, so the equality \eqref{eq-optimize} holds. From now, let us focus on the case $2k>d$. First, we consider an arbitrary orthonormal subset $\set{v_1,\ldots, v_k}\subset \Comp^d$ and orthogonal projections $\Pi_v=\sum_{j=1}^k |v_j\ra\la v_j|$ and $\Pi_{\overline{v}}=\sum_{j=1}^k |\overline{v_j}\ra \la \overline{v_j}|$. Then we can observe that
    \begin{align*}
        \sum_{j,j'} |\la {v_j}|\overline{v_{j'}}\ra|^2=\text{Tr}(\Pi_v \Pi_{\overline{v}})&\geq \text{dim}(\text{Ran}(\Pi_{v})\cap \text{Ran}(\Pi_{\overline{v}}))
        \end{align*}
        and the right hand side is equal to
        \begin{align*}
        \text{dim}(\text{Ran}(\Pi_{v}))+\text{dim}(\text{Ran}(\Pi_{\overline{v}}))-\text{dim}(\text{Ran}(\Pi_{v})+\text{Ran}(\Pi_{\overline{v}})).
    \end{align*}
 Thus we have $\sum_{j,j'} |\la {v_j}|\overline{v_{j'}}\ra|^2 \geq 2k-d$. On the other hand, let us take a specific orthonormal subset $\left\{v_1,v_2,\cdots,v_k\right\}\subseteq \Comp^d$ where $|v_j\ra=\frac{1}{\sqrt{d}}\sum_{l=1}^d\om^{l(j-1)}|l\ra$ ($j=1,\ldots, k$) and $\om=\exp(\frac{2\pi i}{d})$. Then we can check that the desired equality $\sum_{j,j'}|\la {v_j}|\overline{v_{j'}}\ra|^2=2k-d$ holds from the relation
    $$\la {v_j}|\overline{v_{j'}}\ra=\frac{1}{d}\sum_{l=1}^d \om^{-l(j+j'-2)}=\begin{cases} 1 & \text{if $j+j'-2\equiv 0$ mod $d$,} \\ 0 & \text{otherwise}.\end{cases}$$
%Therefore, we again have \eqref{eq-optimize}.
\end{proof}

\begin{proof}[\textbf{Proof of Theorem \ref{thm-OOkpos} (2) and (3)}]
For an orthonormal subset $\set{v_1,\ldots, v_k}$ in $\Comp^d$, the associated bipartite matrix in \eqref{eq40} is given by
    $$C_k^v(\Le_{p,q})=\frac{1-p-q}{d}I_k\otimes I_d+pk|\Om_k^v\ra\la \Om_k^v|+qF_k^v,$$
where $\begin{cases} |\Om_k^v\ra=\frac{1}{\sqrt{k}}\sum_{j=1}^k|j\ra\otimes |v_j\ra\in \Comp^k\otimes \Comp^d,\\
F_k^v=\sum_{i,j=1}^k|i\ra\la j|\otimes |\overline{v_j}\ra\la \overline{v_i}|\in \M{k}\otimes \M{d}.\end{cases}$ 
%($\overline{v_i}$ is the complex conjugate vector of $v_i$).
Moreover, we can write $F_k^v=\Pi_{\mathcal{S}}^v-\Pi_{\mathcal{A}}^v$, where $\Pi_{\mathcal{S}}^v$ is the projection onto the (symmetric) space ${\rm span}\set{\frac{|i\ra|\overline{v_j}\ra+|j\ra|\overline{v_i}\ra}{\sqrt{2}}: 1\leq i\leq j\leq k}$ and $\Pi_{\mathcal{A}}^v$ is the projection onto the (anti-symmetric) space ${\rm span}\set{\frac{|i\ra|\overline{v_j}\ra-|j\ra|\overline{v_i}\ra}{\sqrt{2}}: 1\leq i< j\leq k}$. Note that ${\rm Ran}(\Pi_{\mathcal{S}}^v)\perp {\rm Ran}(\Pi_{\mathcal{A}}^v)$, and $|\Om_k^v\ra\perp {\rm Ran}(\Pi_{\mathcal{A}}^v)$ since
    $$\left \la \Om_k^v\Bigg|\frac{e_i\otimes \overline{v_j}-e_j\otimes \overline{v_i}}{\sqrt{2}}\right \ra=\frac{1}{\sqrt{2k}}(\la v_i|\overline{v_j}\ra- \la v_j|\overline{v_i}\ra)=0,$$
    for all $1\leq i<j\leq k$. Therefore, after rewriting
    $$C_k^v(\Le_{p,q})=\left(A(I_{kd}-\Pi_{\mathcal{A}}^v)+pk|\Om_k^v\ra\la \Om_k^v|+q\Pi_{\mathcal{S}}^v\right)\oplus \left(A-q\right)\Pi_{\mathcal{A}}^v$$
with $A:=\frac{1-p-q}{d}$ for an orthogonal decomposition, the desired condition $C_k^v(\Le_{p,q})\in \pos_{kd}$ is equivalent to $A-q\geq 0$ and 
\begin{equation}\label{eq41}
    A(I_{kd}-\Pi_{\mathcal{A}}^v)+pk|\Om_k^v\ra\la \Om_k^v|+q\Pi_{\mathcal{S}}^v\in \pos
\end{equation} 
(the condition $k\geq 2$ ensures that $\Pi_{\mathcal{A}}^v$ is nonzero). 

A technical difficulty on \eqref{eq41} is that $|\Om_k^v\ra\la \Om_k^v|$ and $\Pi_{\mathcal{S}}^v$ are not simultaneously diagonalizable since
\[|\Om_k^v\ra\la \Om_k^v| \cdot \Pi_{\mathcal{S}}^v \neq \Pi_{\mathcal{S}}^v \cdot |\Om_k^v\ra\la \Om_k^v| \] 
in general unless $\set{v_i}_{i=1}^k\subset \Real^d$. To overcome this, let us take $|\xi_1\ra=\Pi_{\mathcal{S}}^v|\Om_k^v\ra\in {\rm Ran}(\Pi_{\mathcal{S}}^v)$ and consider $|\Om_k^v\ra=|\xi_1\ra+|\xi_2\ra$. Then 
\[\xi_2\perp ({\rm Ran}(\Pi_{\mathcal{S}}^v)\cup {\rm Ran}(\Pi_{\mathcal{A}}^v)).\]
Moreover, we have the following block matrix decomposition
\begin{align} \label{eq-OOblock}
    &A(I_{kd}-\Pi_{\mathcal{A}}^v)+pk|\Om_k^v\ra\la \Om_k^v|+q\Pi_{\mathcal{S}}^v \nonumber \\
    &\cong\footnotesize\begin{pmatrix} (A+q)\Pi_{\mathcal{S}}^v+pk|\xi_1\ra\la \xi_1| & pk|\xi_1\ra\la \xi_2| \\ pk|\xi_2\ra\la \xi_1| & A(I-\Pi_{\mathcal{S}}^v-\Pi_{\mathcal{A}}^v)+pk|\xi_2\ra\la \xi_2|\end{pmatrix}.
\end{align}
An important fact to note on \eqref{eq-OOblock} is that  ${\rm rank}(\Pi_{\mathcal{S}}^v)=\frac{k(k+1)}{2}\geq 2$ and ${\rm rank}(I-\Pi_{\mathcal{S}}^v-\Pi_{\mathcal{A}}^v)=d^2-k^2\geq 2$ since $1<k<d$. Therefore, the block matrix in \eqref{eq-OOblock} is positive semidefinite if and only if
\begin{equation} \label{eq-OOblock2}
    \begin{cases}
        (a)& A+q\geq 0,\\
        (b)& A\geq 0,\\
        (c)& A+q+pk\|\xi_1\|^2\geq 0,\\
        (d)& A+pk\|\xi_2\|^2\geq 0,\\
        (e)& \left(A+q+pk\|\xi_1\|^2\right)\left(A+pk\|\xi_2\|^2\right)\geq (pk)^2\|\xi_1\|^2\|\xi_2\|^2.
    \end{cases}
\end{equation}
Since $\|\xi_1\|^2+\|\xi_2\|^2=\|\Om_k^v\|^2=1$, the conditions $(d)$ and $(e)$ can be understood as
\begin{equation}
    \begin{cases}
        (d')& A+pk-pk\|\xi_1\|^2\geq 0,\\
        (e')& (A+q)(A+pk)-pkq\|\xi_1\|^2\geq 0.
    \end{cases}
\end{equation}
Note that the first two conditions (a) and (b) are independent of the choices of an orthonormal subset $\set{v_i}_{i=1}^k\subset \Comp^d$, and that  the other inequalities in $(c)$, $(d')$ and $(e')$ are linear in $\|\xi_1\|^2$. Since the inequalities in $(c)$, $(d')$ and $(e')$ should hold for all possible choices of $\set{v_i}_{i=1}^k\subseteq \Comp^d$, it suffices to calculate the maximum and minimum values of $\|\xi_1\|^2$.

Recall that $|\Om_k^v\ra=|\xi_1\ra\in {\rm Ran}(\Pi_{\mathcal{S}}^v)$ whenever $\set{v_i}_{i=1}^k\subset \Real^d$ (choose $|v_i\ra=|i\ra$ for example), so the maximum of $\|\xi_1\|^2$ is 1. For the minimum of $\|\xi_1\|^2$, the following expression of $\|\xi_1\|^2$ in terms of $v_1,v_2,\cdots,v_k$
\begin{align*}
    \|\xi_1\|^2&=\la \Om_k^v|\Pi_{\mathcal{S}}^v|\Om_k^v\ra\\
    &=\la \Om_k^v|F_k^v|\Om_k^v\ra \quad (\,\because |\Om_k^v\ra\perp {\rm Ran}(\Pi_{\mathcal{A}}^v))\\
    &=\frac{1}{k}\sum_{j,j'}\la v_j|\overline{v_{j'}}\ra \cdot \la \overline{v_{j'}}|v_j\ra=\frac{1}{k}\sum_{j,j'}|\la v_j|\overline{v_{j'}}\ra|^2.
\end{align*}
allows us to apply  Lemma \ref{lem-OOoptimization} to conclude that $\min \|\xi_1\|^2=\frac{\max(2k-d,0)}{k}$.

To summarize, $\Le_{p,q}\in \PO_k$ if and only if the six inequalities $(a)$, $(b)$, $(c)$, $(d')$, $(e')$, and $A-q\geq 0$ hold for $\|\xi_1\|^2\in \set{1, m:=\frac{\max(2k-d,0)}{k}}$ and $A=\frac{1-p-q}{d}$. 
For the cases $1<k\leq d/2$, we have $m=0$ and obtain the inequalities in (2). Also, for the cases $d/2<k<d$, we have $m=\frac{2k-d}{k}$ and obtain the inequalities in (3). In particular, the inequality $f_k(x,y)\leq 0$ is coming from $(e')$ with $\|\xi_1\|^2=\frac{2k-d}{k}$.
\end{proof}

\subsection{Schmidt numbers of orthogonally invariant quantum states}

Now we are almost ready to compute the Schmidt numbers of all quantum states of the form
\begin{equation} \label{eq-OOInv2}
    \rho_{a,b}:=\frac{1-a-b}{d^2}I_d\otimes I_d+a|\Om_d\ra\la \Om_d|+\frac{b}{d}{F_d}
\end{equation}

Let us denote by $S_k:=\set{(a,b)\in \Real^2:\rho_{a,b}\in \sch_k}$. The main aim of this Section is to prove $S_1\subsetneq S_2\subsetneq \cdots \subsetneq S_d$ with both geometric and algebraic descriptions. Our strategy is to combine Theorem \ref{cor-invSch} and the explicit descriptions of $P_k=\set{(p,q)\in \Real^2:\Le_{p,q}\in \PO_k}$ (Theorem \ref{thm-OOkpos}). 

First of all, Theorem \ref{cor-invSch} implies that we have 
\begin{align}
&{\rm SN}(\rho_{a,b})\leq k \nonumber \\
& \iff \la \Om_d|(\id_d\otimes  \Le_{p,q})(\rho_{a,b})|\Om_d\ra\geq 0 \text{ for all }(p,q)\in {\rm Ext}(P_k) \nonumber\\
    & \iff \begin{pmatrix} p & q \end{pmatrix}\begin{pmatrix}d+1 & 1 \\ 1 & d+1\end{pmatrix}\binom{a}{b}\geq -\frac{1}{d-1} \text{ for all }(p,q)\in {\rm Ext}(P_k) \label{eq-OOparing}. 
\end{align}

Let us denote by 
\begin{equation}\label{eq43}
H_{p,q}=\left\{(x,y)\in \mathbb{R}^2: px+qy\leq 1 \right\}
\end{equation}
for all $(p,q)\in \mathbb{R}^2\setminus\left\{(0,0)\right\}$, and by $\alpha:\mathbb{R}^2\rightarrow \mathbb{R}^2$ a linear isomorphism given by 
\begin{equation}\label{eq44}
    \alpha: \left (\begin{array}{cc} x\\y \end{array} \right ) \mapsto -(d-1)\left ( \begin{array}{cc}
d+1&1\\
1&d+1
\end{array} \right ) \left (\begin{array}{cc} x\\y \end{array} \right ).
\end{equation}
Then we have the following identity:
\begin{equation}\label{eq42}
    S_k=\bigcap_{(p,q)\in {\rm Ext}(P_k)}H_{\alpha(p,q)}.
\end{equation}

This is where a detailed geometric analysis of $P_k$ (Theorem \ref{thm-OOkpos}) manifests its efficacy. Indeed, Theorem \ref{thm-OOkpos} states that the geometric structures of $P_k$ are categorized into four distinct cases $\left\{\begin{array}{llll}(1)&k=1\\ 
(2) & 1<k\leq \frac{d}{2}\\
(3) & \frac{d}{2}<k<d\\
(4) & k=d
    \end{array} \right .$. Furthermore, for the three cases (1), (2), (4), the associated regions $P_k$ are compact convex sets with at most four extreme points. Thus, it is enough to use at most four Schmidt number witnesses $\Le_{p,q}$ to determine $S_k$ by Theorem \ref{cor-invSch} and \eqref{eq42}, and the consequence is that $S_k$ is an intersection of at most four closed half-planes for the three cases (1), (2), (4). All our discussions above are summarized into the following theorem.

\begin{theorem} \label{thm-OOSch}
Let $\rho_{a,b}$ be a bipartite matrix of the form \eqref{eq-OOInv} and $1\leq k\leq d$. Then we have
    $$S_k=\bigcap_{(p,q)\in {\rm Ext}(P_k)}H_{\alpha(p,q)}$$
where $\alpha$ and $H_{p,q}$ are from \eqref{eq43} and \eqref{eq44}. Moreover, we have the following algebraic descriptions for the three cases $\left\{\begin{array}{llll}(1)&k=1,\\ 
(2) & 1<k\leq \frac{d}{2},\\
(4) & k=d.
\end{array} \right .$
\begin{enumerate}
    \item [(1)]$\rho_{a,b}\in \SEP$ if and only if $\begin{cases} -\frac{1}{d-1}\leq (d+1)a+b\leq  1, \\ -\frac{1}{d-1}\leq a+(d+1)b\leq 1. \end{cases}$
    
    \item [(2)]$\rho_{a,b}\in \sch_k$ ($1<k\leq \frac{d}{2}$) if and only if $\begin{cases} -\frac{1}{d-1}\leq (d+1)a+b\leq  \frac{kd-1}{d-1}, \\ a+(d+1)b\leq 1, \\ -\frac{d-k+1}{kd+k-1}a+b\geq -\frac{1}{d-1}.\end{cases}$
    
    \item [(4)]$\rho_{a,b}\in \sch_d=\pos$ if and only if $\begin{cases} a-(d-1)b\leq 1, \\ a+(d+1)b\leq 1, \\ (d+1)a+b\geq -\frac{1}{d-1}.\end{cases}$ 
\end{enumerate}
\end{theorem}

We should remark that the remaining case (3) is quite different from the other cases since there are infinitely many extreme points in $P_k$. In this case, we will utilize some elementary geometric tools from projective geometry to overcome the technical issue.  Indeed, we need a quadratic curve to describe $S_k$ for the cases $\frac{d}{2}<k<d$. This excluded case (3) will be discussed with details independently in Subsection \ref{sec-SN-hard}.

Although we postpone the proof of the remaining case $\frac{d}{2}<k<d$ to Subsection \ref{sec-SN-hard}, let us exhibit a visualized geometric structures of $S_1$, $S_2$, $\cdots$, $S_d$ in the following Figure \ref{fig-OOSch}, particularly for the cases $d=3$ and $d=4$.

\begin{figure}[htb!]
    \centering
    \includegraphics[scale=0.31]{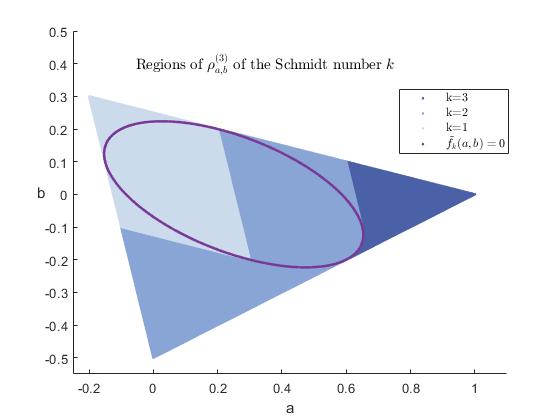} \includegraphics[scale=0.31]{230714_4.jpg}
    \caption{Regions of $\rho^{(d)}_{a,b}$ of the Schmidt number $k$ for $d=3,4$}
    \label{fig-OOSch}
\end{figure}

As in the case of $k$-positivity of $\Le_{p,q}$, the geometric structures of the convex subsets $S_k$ can be categorized into the following four distinct cases.

\vspace{4pt}

\begin{enumerate}
    \item The region $S_1$ is \textit{rhombus-shaped} with vertices $(-\frac{2}{d^2+d-2},\frac{d}{d^2+d-2})$, $(\frac{1}{d+2}, \frac{1}{d+2})$, $(\frac{d}{d^2+d-2},-\frac{2}{d^2+d-2})$, and $(-\frac{1}{d^2+d-2},-\frac{1}{d^2+d-2})$.

    \vspace{3pt}
    
    \item If $1<k\leq \frac{d}{2}$, then the region $S_k$ is \textit{trapezoid-shaped} with vertices $(-\frac{2}{d^2+d-2},\frac{d}{d^2+d-2})$, $(\frac{kd+k-2}{d^2+d-2}, \frac{d-k}{d^2+d-2})$, $(\frac{kd+k-1}{d^2+d-2},-\frac{k+1}{d^2+d-2})$, and $(0,-\frac{1}{d-1})$.

    \vspace{3pt}
    
    \item If $\frac{d}{2}<k<d$, then the region $S_k$ is bounded by a piecewise-linear curve joinig $(\frac{d}{3d-2k}, -\frac{2d-2k}{(d-1)(3d-2k)})$, $(0,-\frac{1}{d-1})$, $(-\frac{2}{d^2+d-2},\frac{d}{d^2+d-2})$, $(\frac{kd+k-2}{d^2+d-2}, \frac{d-k}{d^2+d-2})$ and  $(\frac{k^2d+k^2+d-3k}{k(d^2+d-2)}, -\frac{(d-k+1)(d-k)}{k(d^2+d-2)})$ in that order, and then joined smoothly by an ellipse from $(\frac{k^2d+k^2+d-3k}{k(d^2+d-2)}, -\frac{(d-k+1)(d-k)}{k(d^2+d-2)})$ to $(\frac{d}{3d-2k}, -\frac{2d-2k}{(d-1)(3d-2k)})$.

    \vspace{3pt}
    
    \item The region $S_d$ is the same with $P_d=\set{(a,b): \Le_{a,b}\in \CP}$, i.e. $S_d$ is a \textit{triangle} with vertices $(1,0), (0,-\frac{1}{d-1})$, and $(-\frac{2}{d^2+d-2}, \frac{d}{d^2+d-2})$.
\end{enumerate}

\subsubsection{Algebraic descriptions of $S_k$ for the cases $\frac{d}{2}<k<d$}\label{sec-SN-hard}

Let us focus on explicit algebraic descriptions of 
\begin{equation} \label{eq-OOSch}
    S_k=\bigcap_{(p,q)\in {\rm Ext}(P_k)}H_{\alpha(p,q)}=\alpha^{-1}\bigg(\bigcap_{(p,q)\in {\rm Ext}(P_k)}H_{p,q}\bigg)
\end{equation} 
for the cases $\frac{d}{2}<k<d$. 
% (here the relation $H_{\alpha(p,q)}=\alpha^{-1}(H_{p,q})$ for each $(p,q)$ justifies the second equality). 
In this case, we have
\footnotesize
    $${\rm Ext}(P_k)=\set{\left(\frac{-2}{d^2+d-2},\frac{d}{d^2+d-2}\right), (1,0), \left(0,\frac{-1}{d-1}\right), \left(\frac{-1}{kd-1},0\right)}\,\cup\, C_k$$
\normalsize
by Theorem \ref{thm-OOkpos} (3). Here, $C_k$ is a conic arc in the second quadrant and is parametrized by a smooth, regular, and strictly convex curve
\begin{equation} \label{eq-conic-Ck}
    \gamma:[0,1]\rightarrow C_k
\end{equation} 
satisfying $f_k(\gamma(t))\equiv 0$, $\gamma(0)=(-\frac{1}{kd-1},0)$ and $\gamma(1)=(-\frac{2}{d^2+d-2},\frac{d}{d^2+d-2})$. Then \eqref{eq-OOSch} implies that $\tilde{S}_k:=\alpha(S_k)$ is given by
    \begin{align}
    \label{eq45}&H_{-\frac{2}{d^2+d-2},\frac{d}{d^2+d-2}}\cap H_{1,0}\cap H_{0,-\frac{1}{d-1}}\cap H_{-\frac{1}{kd-1},0}\cap \bigcap_{t\in [0,1]}H_{\gamma(t)},
    \end{align}
and the most technical problem is to demonstrate that $\displaystyle \bigcap_{t\in [0,1]}H_{\gamma(t)}$ is a convex set bounded by two lines and one conic arc as in the following Figure \ref{fig-intersection1}.

    \begin{figure}[htb!]
    \centering
    \includegraphics[scale=0.55]{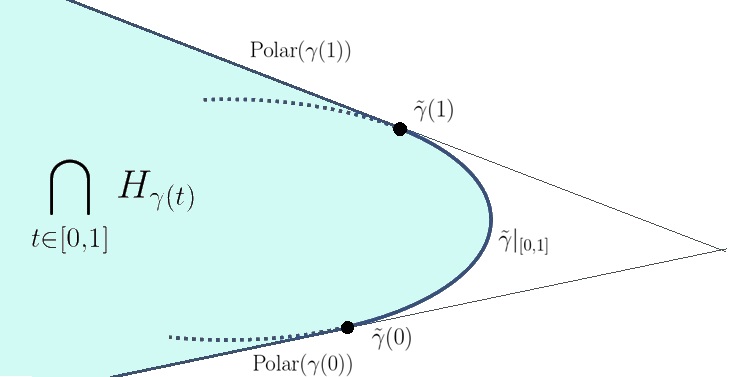}
    \caption{Geometric description of the intersection $\displaystyle \bigcap_{t\in [0,1]} H_{\gamma(t)}$}
    \label{fig-intersection1}
\end{figure}

Here, we need to explain the {\it dual curve} and the {\it pole-polar duality} from projective geometry \cite{BK86, Cox03}. Firstly, we have an explicit formula for the dual curve $\tilde{\gamma}$ of a strictly convex smooth curve $\gamma$. Recall that a plane curve $\gamma: I\to \Real^2$ defined on an open interval $I$ is called \textit{strictly convex} if the number of intersection points between $\gamma$ and an arbitrary line is at most $2$. If $\gamma$ is smooth, then the strict convexity of $\gamma$ is equivalent to that for every $t\in I$, the image of $\gamma$ is contained in the same half-plane whose boundary is the tangent line $l_t$ at $\gamma(t)$, and $\gamma(t)$ is the unique intersection point of $l_t$ and $\gamma$. For a strictly convex smooth curve $\gamma=(p(t),q(t))$ (with additional conditions in Lemma \ref{lem-curve}), we define its dual curve $\tilde{\gamma}$ by
\begin{equation}\label{eq-dual-curve}
    \tilde{\gamma}(t):=\left(\frac{q'(t)}{p(t)q'(t)-q(t)p'(t)}, \frac{-p'(t)}{p(t)q'(t)-q(t)p'(t)}\right).
\end{equation}

Secondly, the pole-polar duality is a bijective correspondence between $\mathbb{R}^2\setminus \left\{(0,0)\right\}$ and the set of all lines that are not passing through the origin $(0,0)$ in $\mathbb{R}^2$. Associated to $(p,q)\in \mathbb{R}^2\setminus \left\{(0,0)\right\}$ is a line
$$l=\left\{(x,y)\in \mathbb{R}^2: px+qy=1\right\}.$$
In this case, we call $l$ the \textit{polar} of $P$ and $P$ the \textit{pole} of $l$ (with respect to the unit circle $C: x^2+y^2=1$), and denote by $l=:{\rm Polar}(P)$ and $P=:{\rm Pole}(l)$ respectively. Note that we have
    $$\tilde{\gamma}(t)={\rm Pole}(l_t)$$ 
for all $t\in I$, where $l_t$ denotes the line tangent to $\gamma$ at $\gamma(t)$.

The following Lemma \ref{lem-curve} establishes the connection between the dual curve $\tilde{\gamma}$ and the intersection $\bigcap_{(p,q)\in C_k}H_{p,q}$. This seems a well-known fact, but we provide a proof for readers' convenience.

\begin{lemma} \label{lem-curve}
Let $I$ be an open interval and $\gamma: I\to \Real^2\setminus \set{(0,0)}$ be a smooth, regular, and strictly convex curve. Suppose that for every $t\in I$, the tangent lines $l_t$ at $\gamma(t)$ do not pass through the origin $(0,0)$, and the origin is in the same (closed) half-plane with $\gamma$ with respect to $l_t$. Then the dual curve $\tilde{\gamma}:I\to \Real^2\setminus\set{(0,0)}$ from \eqref{eq-dual-curve} satisfies the following properties.
\begin{enumerate}
    \item ${\rm Polar}(\gamma(t))$ is tangent to $\tilde{\gamma}$ at $\tilde{\gamma}(t)$ for each $t\in I$. %In other words, $\tilde{\gamma}$ is an {\rm envelope} of the family $\set{{\rm Polar}(\gamma(t))}_{t\in I}$.

    \item $\tilde{\gamma}$ is smooth and strictly convex, and the origin is in the same half-plane with $\widetilde{\gamma}$ with respect to ${\rm Polar}(\gamma(t))$.

    \item For any closed interval $[t_0,t_1]\subset I$, the intersection
        $$\bigcap_{t\in [t_0,t_1]} H_{\gamma(t)}$$
    is the largest convex region containing $(0,0)$, which is bounded by two lines ${\rm Polar}(\gamma(t_0))$ and ${\rm Polar}(\gamma(t_1))$ as well as the dual curve $\tilde{\gamma}|_{[t_0,t_1]}$ (as in Figure \ref{fig-intersection1}).

    \item If $\gamma$ represents a connected part of a conic, then so is $\tilde{\gamma}$.
\end{enumerate}
\end{lemma}
\begin{proof}
Set $\gamma(t)=(p(t),q(t))$. Then the equation of $l_t$ is given by
\begin{equation} \label{eq-lt}
    q'(t)(x-p(t))-p'(t)(y-q(t))=0.
\end{equation}
Note that the given assumptions imply that $p(t)q'(t)-q(t)p'(t)\neq 0$ for all $t\in I$, so by continuity, we may assume $p(t)q'(t)-q(t)p'(t)>0$ for all $t\in I$ without loss of generality. In particular, the dual curve $\widetilde{\gamma}$ from \eqref{eq-dual-curve} is a well-defined smooth curve, and \eqref{eq-lt} implies that $\tilde{\gamma}(t)={\rm Pole}(l_t)$ for all $t\in I$. Let us write $\tilde{\gamma}(t)=(\tilde{x}(t),\tilde{y}(t))$ for simplicity and explain why the four conclusions (1)-(4) hold.

(1) It is enough to check that
\begin{align*}
    p(t)\tilde{x}(t)+q(t)\tilde{y}(t)&=1,\\
    p(t)\tilde{x}'(t)+q(t)\tilde{y}'(t)&=0.
\end{align*}
Indeed, the first equation comes from the fact that $\tilde{\gamma}(t)={\rm Pole}(l_t)$, and the second equation is obtained by differentiating the first equation and the identity $p'(t)\tilde{x}(t)+q'(t)\tilde{y}(t)\equiv 0$ from \eqref{eq-dual-curve}.

(2) Note that the strict convexity of $\gamma$ implies that
\begin{equation} \label{eq-strictconvex}
    q'(t)(p(s)-p(t))-p'(t)(q(s)-q(t))<0, \quad t\neq s\in I,
\end{equation}
which is equivalent to
\begin{equation} \label{eq-strictconvex2}
    p(s)\tilde{x}(t)+q(s)\tilde{y}(t)<1, \quad t\neq s\in I.
\end{equation}
Thus, both the origin and $\tilde{\gamma}$ are in the same half-plane with respect to ${\rm Polar}(\gamma(s))$. Moreover, \eqref{eq-strictconvex2} implies that $\tilde{\gamma}$ is strictly convex, and smoothness is immediate from the explicit description \eqref{eq-dual-curve} of $\tilde{\gamma}$.

(3) Let $T$ be the largest convex region containing $(0,0)$ bounded by two lines ${\rm Polar}(\gamma(t_0))$, ${\rm Polar}(\gamma(t_1))$, and the dual curve $\tilde{\gamma}|_{[t_0,t_1]}$. First, it is immediate to see that $T\subseteq \bigcap_{t\in [t_0,t_1]} H_{\gamma(t)}$. Indeed, $T$ should be contained in the same plane with $(0,0)$ with respect to each tangent line $l_t={\rm Pole}(\gamma(t))$ for all $t\in I$, and this implies $T\subseteq H_{\gamma(t)}$ for all $t\in [t_0,t_1]$. On the other side, let us pick an element $(u,v)\notin T$ and let $l$ be the straight line passing through the origin and $(u,v)$. Then $l$ intersects with one of ${\rm Polar}(\gamma(t_0))$, ${\rm Polar}(\gamma(t_1))$, and $\tilde{\gamma}|_{[t_0,t_1]}$. For the first two cases, $(u,v)$ and $(0,0)$ are not on the same half-plane with respect to either ${\rm Polar}(\gamma(t_0))$ or ${\rm Polar}(\gamma(t_1))$, so $(u,v)\notin \bigcap_{t\in [t_0,t_1]} H_{\gamma(t)}$. For the remaining case, if we suppose that $l$ contains certain $\tilde{\gamma}(t)$, then $(0,0)$ and $(u,v)$ are not on the same plane with respect to $H_{\gamma(t)}$. Hence, we can conclude that $T^c\subseteq \left (\bigcap_{t\in [t_0,t_1]} H_{\gamma(t)}\right )^c$, i.e. $\bigcap_{t\in [t_0,t_1]} H_{\gamma(t)}\subseteq T$.

(4) This is a direct consequence from Pl\"{u}cker's formula \cite[Section 9.1]{BK86}, which states that the degree of the dual curve $\tilde{\gamma}$ is $n(n-1)$ for any non-singular plane algebraic curve $\gamma$ of degree $n$ (in our case, $n=2$). Alternatively, more elementary arguments can be found in \cite{CG96, AZ07}.
\end{proof}

From now on, let us focus more on the special case $\gamma: [0,1]\to C_k$ from \eqref{eq-conic-Ck}. We may assume that $\gamma$ is extended to the smooth, regular, and strictly convex curve (still denoted by $\gamma$) on an open interval $I\supset [0,1]$ such that $f_k\circ \gamma\equiv 0$. Then Lemma \ref{lem-curve} (4) implies that there exists a quadratic polynomial $\tilde{f}_k(x,y)$ such that 
$$\tilde{f}_k\left ((\alpha^{-1} \circ \tilde{\gamma})(t)\right )\equiv 0.$$ 
Here, $\tilde{\gamma}$ is the dual curve of $\gamma$, and $\alpha$ is the linear isomorphism from \eqref{eq44}. A notable fact is that $\alpha^{-1}\circ \tilde{\gamma}$ always represents an ellipse. For this conclusion, the following lemma provides more concrete information on the quadratic polynomial $\tilde{f}_k(x,y)$.

\begin{lemma}\label{lem41}
The quadratic equation $\tilde{f}_k(x,y)=0$ holds for the following five points $(a_i,b_i)$ $(1\leq i\leq 5)$
    \begin{center}
      \small  $\left(-\frac{d}{k(d^2+d-2)}, \frac{d^2-kd+d-2k}{k(d^2+d-2)}\right), \left(\frac{d^2-kd+d-k-1}{d^2+d-2}, -\frac{d-k+1}{d^2+d-2}\right), \left(\frac{2kd-d^2+2k-d-2}{d^2+d-2}, \frac{2d-2k}{d^2+d-2}\right)$, $\left(\frac{k^2d+k^2+d-3k}{k(d^2+d-2)}, -\frac{(d-k+1)(d-k)}{k(d^2+d-2)}\right), \left(\frac{d}{3d-2k}, -\frac{2d-2k}{(d-1)(3d-2k)}\right)$, \normalsize 
    \end{center}
    with the associated tangent lines $l_i$ ($1\leq i \leq 5$)
    \begin{center}
    $(d+1)x+y=-\frac{1}{d-1}$,\quad  $x+(d+1)y=-\frac{1}{d-1}$,\quad  $x+(d+1)y=1$,
    
    $(d+1)x+y=\frac{kd-1}{d-1}$, \quad $x-(d-1)y=1$, 
\end{center}   
respectively. Furthermore, if $\frac{d}{2}<k<d$, the conic determined by the equation $\tilde{f}_k(x,y)=0$ is inscribed in the convex pentagon bounded by the above five tangent lines. In particular, the equation $\tilde{f}_k(x,y)=0$ should represent an ellipse.
\end{lemma}
\begin{proof}

%Recall that each $\gamma(t)=:(p,q)$ on the conic $f_k(x,y)=0$ is associated to $\alpha(\tilde{\gamma}(t))=:(a,b)$ on the conic $\tilde{f}_k(x,y)=0$ by Lemma \ref{lem-curve}, and we need to know $\gamma'(t)$ for an explicit description 

Let us begin with the following expression
\begin{equation}\label{eq48}
    f_k(x,y)=Ax^2+Bxy+Cy^2+Dx+Ey+F
\end{equation} 
with the coefficients
    $$\begin{cases}A=kd-1, \quad B=-(d^3-kd^2-kd-d+2), \quad C=d-1, \\ D=-kd+2, \quad E=-d+2, \quad  F=-1,\end{cases}$$
from \eqref{eq-conic1}.
% Let $\gamma: I\rightarrow \mathbb{R}^2$ be a smooth regular strictly convex curve satisfying $f_k\circ \gamma \equiv 0$, and let us denote by $\gamma(t)=(p(t),q(t))$. Note that the quadratic polynomial $\tilde{f}_k$ satisfies $\tilde{f}_k\circ \alpha^{-1} \circ \tilde{\gamma}\equiv 0$, so all $\alpha^{-1}(\tilde{\gamma}(t))$ are solutions of the equation $\tilde{f}_k(a,b)=0$. Moreover,  \eqref{eq48} allows us to compute $\alpha^{-1}\circ \tilde{\gamma}$ explicitly. Indeed,
If we write $\gamma(t)=(p(t),q(t))$, then we have
    \begin{equation}\label{eq47}
        \tilde{\gamma}(t)=-\left(\frac{2Ap(t)+Bq(t)+D}{Dp(t)+Eq(t)+2F}, \frac{2Cq(t)+Bp(t)+E}{Dp(t)+Eq(t)+2F}\right)
    \end{equation}
thanks to \eqref{eq48} and \eqref{eq44}.
% of the linear isomorphism $\alpha$ from \eqref{eq44}. 

Recall that $\alpha^{-1}(\tilde{\gamma}(t))$ are solutions of the equation $\tilde{f}_k(x,y)=0$ for all $t\in I$. Thus, in order to single out five points $(a_i,b_i)$ satisfying $\tilde{f}_k(a_i,b_i)=0$, it is enough to note that the following five points $(p(t_i),q(t_i))$ ($1\leq i\leq 5$)
    \[(1,0), (0,1), \left(0,\frac{-1}{d-1}\right), \left(\frac{-1}{kd-1},0\right), \left(\frac{-2}{d^2+d-2},\frac{d}{d^2+d-2}\right)\]
are solutions to the equation $f_k(x,y)=0$. Then \eqref{eq47} provides us with the associated five points $(a_i,b_i)=\alpha^{-1}(\tilde{\gamma}(t_i))$ listed in the statement. Furthermore, the tangent lines $l_i$ at $(a_i,b_i)$ satisfying $\tilde{f}_k(a_i,b_i)=0$ are given by ${\rm Polar}(\alpha(p(t_i),q(t_i)))$, by Lemma \ref{lem-curve} (1). Thus, we can write down what the tangent lines are explicitly, as in the statement. Lastly, it is immediate to check that when $\frac{d}{2}<k<d$, those five tangent lines consist of a convex pentagon, and the corresponding points $(a_i,b_i)$ of tangency are on each of the pentagon's sides. This observation forces the quadratic equation $\tilde{f}_k(x,y)=0$ to represent an ellipse inscribed in this pentagon.
\end{proof}

% In particular, Lemma \ref{lem-curve} again implies that the conic $\tilde{f}_k(x,y)=0$ is tangent to two lines $x-(d-1)y=1$ and $(d+1)+y=\frac{kd-1}{d-1}$ at $(\frac{d}{3d-2k}, -\frac{2d-2k}{(d-1)(3d-2k)})$ and $(\frac{k^2d+k^2+d-3k}{k(d^2+d-2)}, -\frac{(d-k+1)(d-k)}{k(d^2+d-2)})$, respectively. Moreover, the conic arc having these two points as endpoints is contained in ${\rm Ext}(S_k)$ while the other part is not.

%Finally, it remains to show that $\tilde{f}_k(x,y)=0$ represents an ellipse. Note that Lemma \ref{lem-curve} (and the linear isomorphism $\alpha$) implies that the five ``canonical'' tangent lines can be given as in Table \ref{tab-ellipse}. If $\frac{d}{2}<k<d$, these lines determine a convex pentagon (containing $S_k$) and the corresponding points of tangency are on each of its sides, meaning that the conic $\tilde{f}_k(x,y)=0$ is inscribed in this pentagon. This observation forces that $\tilde{f}_k(x,y)=0$ is an ellipse.

\begin{remark} \label{rmk-ellipse}
While the dual quadratic equation $\tilde{f}_k(x,y)=0$ in our consideration always represents an ellipse thanks to Lemma \ref{lem41}, the quadratic equation $f_k(x,y)=0$ can represent both an ellipse and a hyperbola. For example, the quadratic equation  $f_k(x,y)=0$ for $d=5$ is given by
    $$(5k-1)p^2-(122-30k)pq+4q^2-(5k-2)p-3q-1=0,$$
and this represents a hyperbola if $k=3$ and an ellipse if $k=4$.
\end{remark}

Finally, we are ready to describe the intersection 
    $$\bigcap_{(p,q)\in C_k}H_{\alpha(p,q)}=\alpha^{-1}\bigg(\bigcap_{(p,q)\in C_k}H_{p,q}\bigg)$$
from \eqref{eq45}. Recall that $\left(\frac{-1}{kd-1},0\right)$ and $\left(\frac{-2}{d^2+d-2},\frac{d}{d^2+d-2}\right)$ are the two end-points $(p,q)$ of the connected conic arc $C_k$, and their associated points $(a,b)$ satisfying $\tilde{f}_k(a,b)=0$ are given by $\left(\frac{k^2d+k^2+d-3k}{k(d^2+d-2)}, -\frac{(d-k+1)(d-k)}{k(d^2+d-2)}\right)$ and $\left(\frac{d}{3d-2k}, -\frac{2d-2k}{(d-1)(3d-2k)}\right)$. Let us denote by $L$ the line segment between these two points, and let us assume (by changing the sign if necessary) that the inequality $\tilde{f}_k(x,y)\leq 0$ represents a filled ellipse.

\begin{corollary}\label{cor41}
Let $(a,b)\in \Real^2$. Then $(a,b)\in \displaystyle \bigcap_{(p,q)\in C_k}H_{\alpha(p,q)}$ if and only if $(a,b)$ satisfies $\tilde{f}_k(a,b)\leq 0$ or satisfies the following three conditions:
\begin{enumerate}
    \item $(d+1)a+b\leq  \frac{kd-1}{d-1}$,
    \item $a-(d-1)b \leq 1$, 
    \item $(3d-k+3)a-(kd+k-3)b-\frac{d^2+kd-k-3}{d-1}\leq 0$.
\end{enumerate}
\end{corollary}

\begin{proof}
   By Lemma \ref{lem-curve} (3) and Lemma \ref{lem41}, the intersection $\displaystyle \bigcap_{(p,q)\in C_k}H_{\alpha(p,q)}$ is the largest convex region containing $(0,0)$ bounded by the two tangent lines $(d+1)x+y=\frac{kd-1}{d-1}$, $x-(d-1)y=1$ and the dual curve $\tilde{\gamma}|_{[0,1]}$. We refer the readers to Figure \ref{fig-intersection1} for a visualized understanding. Note that  $\tilde{f}_k(x,y)\leq 0$ represents a filled ellipse which we denote by $E$, and $E$ is a subset of the intersection $\displaystyle \bigcap_{(p,q)\in C_k}H_{\alpha(p,q)}$ by Lemma \ref{lem41}. Furthermore, $\displaystyle \bigcap_{(p,q)\in C_k}H_{\alpha(p,q)}\setminus E$ is a subset of the largest convex region bounded by the two tangent lines $(d+1)x+y=\frac{kd-1}{d-1}$, $x-(d-1)y=1$ and the line segment $L$ between $\left(\frac{k^2d+k^2+d-3k}{k(d^2+d-2)}, -\frac{(d-k+1)(d-k)}{k(d^2+d-2)}\right)$ and $\left(\frac{d}{3d-2k}, -\frac{2d-2k}{(d-1)(3d-2k)}\right)$. Hence, the conclusion follows immediately.
\end{proof}

Now we are ready to complete the proof for the cases $\frac{d}{2}<k<d$.

\begin{theorem} \label{thm-OOSch2}
Let $\rho_{a,b}$ be a bipartite matrix of the form \eqref{eq-OOInv} and $\frac{d}{2}<k< d$. Then $\rho_{a,b}\in \sch_k$ if and only if $\tilde{f}_k(a,b)\leq 0$ or $(a,b)$ satisfies the following inequalities:
$$\begin{cases}
    -\frac{1}{d-1}\leq (d+1)a+b\leq  \frac{kd-1}{d-1},\\
    a+(d+1)b\leq 1,\\
    a-(d-1)b \leq 1,\\
    (3d-k+3)a-(kd+k-3)b-\frac{d^2+kd-k-3}{d-1}\leq 0.
\end{cases}$$
Here, $\tilde{f}_k$ is the quadratic polynomial from Lemma \ref{lem41} such that the inequality $\tilde{f}_k(x,y)\leq 0$ represents a filled ellipse.
\end{theorem}

\begin{proof}
Since $\displaystyle \bigcap_{(p,q)\in C_k}H_{p,q}\subseteq H_{-\frac{2}{d^2+d-2},\frac{d}{d^2+d-2}}\cap H_{-\frac{1}{kd-1},0}$, we have
    \begin{align}
   \tilde{S}_k:=\alpha(S_k)= H_{1,0}\cap H_{0,-\frac{1}{d-1}}\cap \bigcap_{(p,q)\in C_k}H_{p,q}
    \end{align}
from \eqref{eq45}. Thus, we have 
\begin{align*}
    S_k&=\alpha^{-1}(\tilde{S}_k)=\alpha^{-1}\left (H_{1,0}\right )\cap \alpha^{-1} \left ( H_{0,-\frac{1}{d-1}}\right ) \cap \alpha^{-1}\left ( \bigcap_{(p,q)\in C_k}H_{p,q} \right )\\
    &=H_{-(d^2-1),-(d-1)}\cap H_{1,d+1} \cap \bigcap_{(p,q)\in C_k}H_{\alpha(p,q)}.
\end{align*}
and the conclusion follows immediately from Lemma \ref{lem41} and Corollary \ref{cor41}.

\end{proof}

% \textbf{ADD} Table of list of $\tilde{f}_k$ for each $(d,k)$, $d=3,4,5,6$ and $7,8$ if possible.

% \textbf{ADD} more remark on the shape of $S_k$ (if any), nontriviality, nonintuitivity for example...

% \textbf{ADD} (JK10) $k$-th operator norm $\|\rho_{a,b}\|_{S(k)}$.

% \textbf{ADD} 1-copy distillability of $\rho_{a,b}$.

% \textbf{ADD} comparison with the criterion in \cite[Lemma 1]{TH00}:
% \begin{equation}
%     \rho\in \sch_k\Longrightarrow \max_{\psi:\text{maximally entangled}} \la \psi|\rho|\psi\ra\leq \frac{k}{d}.
% \end{equation}

% \textbf{ADD} Problem in BSGS22.

\begin{remark}
It is worth remarking that a small perturbation can produce a drastic increment of the Schmidt number. Recall that $S_d$ is a triangle, and let us parameterize the southern-eastern edge of $S_d$ by $\eta:[0,1]\rightarrow S_d$ such that $\eta(0)=(0,-\frac{1}{d-1})$ and $\eta(1)=(1,0)$. Then $\rho_{\eta(t)}$ is always entangled, and the Schmidt numbers ${\rm{SN}}(\rho_{\eta(t)})$ exhibit a monotonically increasing pattern  of 
\begin{center}
    $2,\lceil \frac{d}{2}\rceil, \lceil \frac{d}{2}\rceil+1,\lceil \frac{d}{2}\rceil+2,\cdots, d,$
\end{center} 
as $t$ increases from $0$ to $1$. Note that there is a huge gap between $2$ and $\lceil \frac{d}{2}\rceil$, which seems entirely new and highly non-trivial. This phenomenon does not appear on the other line segments in the boundary of $S_d$, and some other known cases such as $\rho_{a,0}$ and $\rho_{0,b}$. The only known patterns were $1,2,3,\cdots,d$ (isotropic states) or $1,2$ (Werner states) to our best knowledge.
\end{remark}

\emph{Acknowledgements}: The authors thank Professor Seung-Hyeok Kye for the helpful discussions and comments. S.-J.Park and S.-G.Youn were supported by the National Research Foundation of Korea (NRF) grant funded by the Ministry of Science and ICT (MSIT) (No.2021K1A3A1A21039365). S.-G.Youn was also supported by Samsung Science and Technology Foundation under Project Number SSTF-BA2002-01 and by the National Research Foundation of Korea (NRF) grant funded by the Ministry of Science and ICT (MSIT) (No. 2020R1C1C1A01009681).

\bibliographystyle{alpha}
\bibliography{Youn23}

\end{document}